\documentclass[english]{lmcs}
\usepackage{amsmath}
\usepackage{amssymb}
\usepackage[T1]{fontenc}
\usepackage[utf8]{inputenc}
\usepackage[english]{babel}
\usepackage{tikz}
\usetikzlibrary{automata,positioning,calc}
\usepackage{ifthen}
\usepackage{hhline}
\newcommand{\Ninfty}{\mathbb{N}_{\infty}}%{\omega+1}
\newcommand{\Nsup}{\infty}%{\omega}
\newcommand{\N}{\mathbb{N}}%{\omega}
\newcommand{\A}{\mathbb{A}}
\newcommand{\B}{\mathbb{B}}

\newcommand{\set}[1]	{\left\{{#1}\right\}}
\newcommand{\perm}[1]{\langle #1\rangle}
\renewcommand{\vec}[1]{\text{\textbf{#1}}}
\newcommand{\notL}{\overline{L}}
\newcommand{\sem}[1]	{[\![#1]\!]}
\newcommand{\nsem}[1]	{[\![#1]\!]_{\neg}}
\newcommand{\semB}[1]	{[\![#1]\!]_B}
\newcommand{\semS}[1]	{[\![#1]\!]_S}
\newcommand{\aut}{\mathcal{A}}

\newcommand{\Fin}{\mathit{Fin}}

\newcommand{\expr}{\mathrm{Expr}}
\newcommand{\oexpr}{\mathrm{Oexpr}}

\newcommand{\coe}{\mathrm{C_{OE}}}

\newcommand{\ose}{$\omega\sharp$-expression}
\newcommand{\se}{$\sharp$-expression}

\newcommand{\intro}[1]{\emph{#1}}
\newcommand{\bowtien}{\underset{n\to\infty}{\bowtie}}%{\omega+1}
\newcommand{\iffb}{\rightleftharpoons}

\newcommand{\ic}{\mathtt{ic}}
\renewcommand{\i}{\mathtt{i}}
\newcommand{\ccr}{\mathtt{cr}}
\renewcommand{\r}{\mathtt{r}}

\newcommand{\e}{\varepsilon}
\newcommand{\val}{\mathrm{val}}

\newcommand{\eval}{\mathrm{eval}}
\newcommand{\cl}[1]{\overline{#1}}
\newcommand{\semi}{\mathbf{S}}
\newcommand{\Smin}{\semi_\mathrm{min}}
\newcommand{\M}{\mathbf{M}}

\newcommand{\ltlq}{\text{LTL}^\leq}
\newcommand{\nltlq}{\text{LTL}^>}
\newcommand{\UN}{U^{\leq N}}
\newcommand{\RgN}{R^{>N}}
\newcommand{\sph}{\sem{\phi}}
\newcommand{\sphapp}{\sem{\phi}^\approx}
\newcommand{\ssp}{\semi_{\sphapp}}

\begin{document}
\title{Linear temporal logic for regular cost functions}
\author{Denis Kuperberg}
\address{{\sc Liafa}/CNRS/Universit\'e Paris 7, Denis Diderot, France}
%\date{\today}
%% mandatory lists of keywords and classifications:
\keywords{LTL, cost functions, cost automata, stabilization semigroup, aperiodic, syntactic congruence}
\subjclass{F.1.1,F.4.3}
\titlecomment{Updated version 08/02/2017}

\begin{abstract}
Regular cost functions have been introduced recently as
an extension to the notion of regular languages with counting
capabilities, which retains strong closure, equivalence, and
decidability properties. The specificity of cost functions is
that exact values are not considered, but only estimated.

In this paper, we define an extension of Linear Temporal Logic (LTL) over finite words to describe cost functions. We give an explicit translation from this new logic to two dual form of cost automata, and we show that the natural decision problems for this logic are PSPACE-complete, as it is the case in the classical setting. We then algebraically characterize the expressive power of this logic, using a new syntactic congruence for cost functions introduced in this paper.
\end{abstract}
\maketitle
\section{Introduction}

Since the seminal works of Kleene and Rabin and Scott, the theory of regular languages is one of the
cornerstones in computer science. Regular languages have many good properties, of
closure, of equivalent characterizations, and of decidability, which makes them
central in many situations.

Recently, the notion of regular cost function for words has been presented as a candidate for being a
quantitative extension to the notion of regular languages, while retaining most of the
fundamental properties of the original theory such as the closure properties, the various equivalent characterizations, and the
decidability \cite{Col09}. A cost function is an equivalence class of
the functions from the domain (words in our case) to $\N\cup\set{\infty}$, modulo an equivalence relation $\approx$
which allows some distortion, but preserves the boundedness property over each subset of the domain.
The model is an extension to the notion of languages in the following sense: one can identify a language
with the function mapping each word inside the language to~$0$, and each word outside the language to~$\Nsup$.
It is a strict extension since regular cost functions have counting capabilities, e.g., counting the number of
occurrences of letters, measuring the length of intervals, etc...

This theory grew out of two main lines of work: research by Hashiguchi \cite{Hashiguchi82}, Kirsten \cite{Kirsten05}, and others who were studying problems which could be reduced to whether or not some function was bounded over its domain (the most famous of these problems being the star height problem); and research by Boja\'{n}czyk and Colcombet \cite{Bojanczyk04,BojanczykC06} on extensions of monadic second-order logic (MSO) with a quantifier $U$ which can assert properties related to boundedness.

Linear Temporal Logic (LTL), which is a natural way to describe logical constraints over a linear structure, have also been a fertile subject of study, particularly in the context of regular languages and automata \cite{VarWol}. Moreover quantitative extensions of LTL have recently been successfully introduced. For instance the model Prompt-LTL introduced in \cite{PromptLTL} is interested in bounding the waiting time of all requests of a formula, and in this sense is quite close to the aim of cost functions.

In this paper, we extend LTL (over finite words) into a new logic with quantitative features ($\ltlq$), in order to describe cost functions over finite words with logical formulae. We do this by adding a new operator $\UN$ : a formula $\phi\UN\psi$ means that $\psi$ holds somewhere in the future, and $\phi$ has to hold until that point, except at most $N$ times (we allow at most $N$ "mistakes" of the Until formula). The variable $N$ is unique in the formula, and the semantic of the formula is the least value of $N$ which makes the statement true.

%Indeed the aim of a $\ltlq$-formula is to bound the number of mistakes of Until operators, and not the waiting time as in Prompt-LTL, so $\ltlq$ is strictly stronger than Prompt-LTL (waiting time can be viewed as consecutive mistakes of formula $\bot$).

%We then characterize the expressive power of $\ltlq$ : as in the case of languages, linear temporal logic corresponds to aperiodic semigroups. 
%This proof is done thanks to a new tool introduced in this paper : syntactic congruence for cost functions.

\subsection*{Related works and motivating examples}

Regular cost functions are the continuation of a sequence of works that intend to solve difficult questions in language theory.
Among several other decision problems, the most prominent example is the star-height problem: given a regular language~$L$ and an integer~$k$,
decide whether~$L$ can be expressed using a regular expression using at most $k$-nesting of Kleene stars.
The problem was resolved by Hashigushi \cite{Hashiguchi88} using a very intricate proof, and later by Kirsten~\cite{Kirsten05} using an automaton that has counting features.
%It was raised by Eggan in~1963~\cite{Eggan63}, but solved only 25 years later by Hashigughi
%using a very intricate proof
%An improved and self-contained proof has been more recently proposed by 
%The two proofs work along the same lines: show that the original problem can be reduced to the
%existence of a bound over some function from words to integers. This function can be
%represented using an automaton that have counting features (a distance automaton for Hashiguchi, and
%a nested distance desert automaton for Kirsten). The proof is concluded by showing that such boundedness
%problems are decidable. 

%Other decision problems can also be reduced to boundedness questions
%over words: in language theory the \emph{finite power property}~\cite{Simon78,Hashiguchi79}
%and the \emph{finite substitution problem}~\cite{Bala04,Kirsten04a},
%and in model theory the \emph{boundedness problem} of monadic formulas
%over words~\cite{BlumensathOttoWeyer09}.
%Distance automata are also used in the context of databases and image
%compression.  Automata similar to the ones of Kirsten have also been
%introduced independently in the context of verification~\cite{AbdullaKY08}.

Finally, also using ideas inspired from \cite{BojanczykC06},
the theory of those automata over words has been unified in~\cite{Col09},
in which cost functions are
introduced, and suitable models of automata, algebra, and logic for
defining them are presented and shown equivalent. Corresponding
decidability results are provided.  The resulting theory is a neat
extension of the standard theory of regular languages to a
quantitative setting. 

On the logic side, Prompt-LTL, introduced in \cite{PromptLTL}, and PLTL \cite{AETP01}, which are similar, show an interesting way to extend LTL in order to look at boundedness issues, and already gave interesting decidability and complexity results. In \cite{DJP04}, the logics $k$TL was introduced, which uses an explicit bound $k$ to express some desired boundedness properties.

These logics are only interested in bounding the wait time, i.e. consecutive events. It would correspond in the framework of regular cost functions to the subclass of temporal cost functions introduced in \cite{CKL}.

We will introduce here a logic $\ltlq$ with a more general purpose : it can bound the wait time before an event, but also non-consecutive events, like the number of occurences of a letter in a word. 

These quantitative issues are a quite natural preoccupation in the context of verification: for instance one would expect that a system can react in a bounded time. The new features of $\ltlq$ could possibly be used to allow some mistakes in the behaviour fo the program, but guarantee a global bound on the number of mistakes. An other issue is the consumption of resources: for instance it is interesting to know whether we can bound the number of times a program stores something in the memory.

\subsection*{Contributions}
It is known from~\cite{Col09} that regular cost functions are the ones recognizable by stabilization semigroups (or in an equivalent way, stabilization monoids), and from \cite{CKL} than there is an effective quotient-wise minimal stabilization semigroup for each regular cost function.
This model of semigroups extends the standard approach for languages.

%The first aim of this paper is to obtain a suitable syntactic congruence for regular cost functions, which allows to get a description of its minimal stabilization semigroup in terms of quotient as it is done in language theory.
%To achieve this goal, we refine the notion of \se~introduced by Hashiguchi in \cite{Has90}, and introduce \ose s which take care of typing constraints of the $\sharp$ operator.

We introduce a quantitative version of LTL in order to describe cost functions by means of logical formulas. The idea of this new logic is to bound the number of "mistakes" of Until operators, by adding a new operator $\UN$. The first contribution of this paper is to give a direct  translation from $\ltlq$-formulae to $B$-automata, which is an extension of the classic translation from LTL to Büchi automaton for languages. This translation preserves exact values (i.e. not only cost functions equivalence), which could be interesting in terms of future applications.
We also use dual forms of logic and cost automata to describe a similar translation, and show that the boundedness problem for $\ltlq$-formulae is PSPACE-complete (as it was the case in the classical setting). Therefore, we do not lose anything in terms of computational complexity, when generalizing from LTL to $\ltlq$.

We then show that regular cost functions described by LTL formulae are the same as the ones computed by aperiodic stabilization semigroups, and this characterization is effective.
The proof uses a syntactic congruence for cost functions, introduced in this paper, which generalizes the Myhill-Nerode equivalence for regular languages. This congruence present a general interest besides this particular context, since it can be used for any regular cost function.

This work validates the algebraic approach for studying cost functions, since it shows that the generalization from regular languages extends also to syntactic congruence. It also allows a more user-friendly way to describe cost functions, since temporal logic is often more intuitive than automata or stabilization semigroups to describe a given cost function.

As it was the case in \cite{CKL} for temporal cost functions, the characterization result obtained here for $\ltlq$-definable cost functions follows the spirit of Schützenberger's theorem, which links star-free languages with aperiodic monoids \cite{Schutz65}.

\subsection*{Organisation of the paper}

After some notations, and reminder on cost functions and stabilization semigroups,  we introduce in Section~\ref{ltl} $\ltlq$ as a quantitative extension of LTL, and  give an explicit translation from $\ltlq$-formulae to $B$ and $S$-automata in Sections \ref{sec:ltlBaut} and \ref{sec:ltlSaut}. We then present in Section \ref{synccong} a syntactic congruence for cost functions,  and show that it indeed computes the minimal stabilization semigroup of any regular cost function. We finally use this new tool to show that $\ltlq$ has the same expressive power as aperiodic stabilization semigroups.

\subsection*{Notations}
%We will reuse notations from \cite{Col09} and \cite{CKL}.
We will note \intro{$\N$} the set of non-negative integers and \intro{$\N_\infty$} the set $\N\cup\{\infty\}$, ordered by~$0<1<\dots<\infty$. We will say that a set $X\subseteq \N_\infty$ is bounded if there is a number $N\in\N$ such that for all $x\in X$, we have $x<N$. In particular, if $X$ contains $\infty$ then $X$ is unbounded.
If $E$ is a set, \intro{$E^\N$} is the set of infinite sequences of elements of~$E$ (we will not use here the notion of infinite word). Such sequences will be denoted by bold letters ($\vec a$, $\vec b$,...).
We will work with a fixed finite alphabet \intro{$\A$}.
The set of words over~$\A$ is \intro{$\A^*$} and the empty word will be noted~\intro{$\epsilon$}. The concatenation of words $u$ and~$v$ is~$uv$. The length of~$u$ is \intro{$|u|$}. The number of occurrences of letter~$a$ in $u$ is~\intro{$|u|_a$}. We will use $|\cdot|$ (resp. $|\cdot|_a$) to note the function $u\mapsto |u|$ (resp. $u\mapsto|u|_a$).
Functions $\N\to\N$ will be denoted by letters $\alpha, \beta,\dots$, and will be extended to $\N\cup\set{\infty}$ by $\alpha(\infty)=\infty$. Such functions will be called \intro{corrections functions}.

\section{Regular Cost functions}\label{costfunc}

\subsection{Cost functions and equivalence}
Let $\mathcal{F}$ be the set of functions from $\A^*$ to $\N_\infty$.
If $L\subseteq\A^*$, we will note $\chi_L$ the function of $\mathcal{F}$ defined by $\chi_L(u)=0$ if $u\in L$, $\infty$ if $u\notin L$.
For $f,g\in\mathcal{F}$, we say that $f\preccurlyeq g$ if for all set $W\subseteq \A^*$, if $g(W)$ is bounded then $f(W)$ is bounded. We define the equivalence relation $\approx$ on $\mathcal{F}$ by $f\approx g$ if $f\preccurlyeq g$ and $g\preccurlyeq f$. Notice that $f\approx g$ means that $f$ and $g$ are bounded on the same sets of words, i.e. for all $W\subseteq\A^*$, we have $f(W)$ is bounded if and only if $g(W)$ is bounded. 
This equivalence relation does not pay attention to exact values, but preserves the existence of bounds.

We also introduce another relation, which is parametrized by a correction function.
If $\alpha$ is a correction function (see Notations), we say that $f\leq_\alpha g$ if $f\leq\alpha\circ g$, and $f\approx_\alpha g $ if $f\leq_\alpha g$ and $g\leq_\alpha f$. Intuitively, $f\approx_\alpha g$ means that one can be obtained from the other by ``distorting'' the value according to the correction function $\alpha$. In particular, $f\approx_\mathit{id} g$ if and only if $f=g$ (where $\mathit{id}$ is the identity function).

\begin{lem}\cite{Col09}
Let $f,g\in\mathcal{F}$. We have $f\preccurlyeq g$ (resp. $f\approx g$) if and only if there exists a correction function $\alpha$ such that $f\leq_\alpha g$ (resp. $f\approx_\alpha g$).
\end{lem}

\begin{proof}
Assume $f\leq_\alpha g$ for some $\alpha$. If $g(W)$ is bounded by $M$ for some set $W\subseteq\A^*$, then $f(W)$ is bounded by $\alpha(M)$, so we get $f\preccurlyeq g$.

Conversely, if $f\preccurlyeq g$, we want to to build $\alpha$ such that $f\leq_\alpha g$.
For each $n\in\N$, we define $W_n=\set{u\in\A^*~|~g(u)\leq n}$.
We define $\alpha(n)=\sup f(W_n)$ if $W_n\neq\emptyset$, and $\alpha(n)=n$ otherwise. As always, $\alpha(\infty)=\infty$.
Notice that because $f\preccurlyeq g$, for every $n\in\N$ we have $\alpha(n)\in\N$, since $f$ is bounded on $W_n$. Let $u\in\A^*$. If $g(u)$ is finite, then let $n=g(u)$, we have $u\in W_n$, so $f(u)\leq\alpha(n)=\alpha\circ g(u)$. If $g(u)=\infty$, then we always have $f(u)\leq\alpha\circ g(u)=\infty$.

We showed that $f\preccurlyeq g$ if and only if there exists a correction function $\alpha$ such that $f\leq_\alpha g$.
It directly follows that if $f\approx_\alpha g$, then $f\preccurlyeq g$ and  $g\preccurlyeq f$, thus $f\approx g$.
Conversely, if $f\approx g$, then there are correction function $\alpha,\beta$ such that $f\leq_\alpha g$ and  $g\leq_\beta f$. We get $f\approx_{\max(\alpha,\beta)} g$.
\end{proof}
Notice that saying $f\approx_\alpha g$ is more precise than saying $f\approx g$: in addition to preserving the qualitative information on bound, the correction function $\alpha$ gives a quantitative information on the distortion of bounds.

A \intro{cost function} is an equivalence class of $\mathcal{F}/{\approx}$. In practice, cost functions will always be represented by one of their elements in $\mathcal{F}$. If $f$ is a function in $\mathcal{F}$, we will note $f^\approx$ the cost function containing $f$. We will say that an object (automaton, logical formula) \intro{recognizes a cost function}, when it defines a function in $\mathcal{F}$, but the notion of equivalence we are mostly interested in is the $\approx$-equivalence instead of the equality of functions.

Notice that the value $\infty$ is considered unbounded, so if $L$ and $L'$ are languages of $\A^*$, then $\chi_L\approx\chi_{L'}$ if and only if $L=L'$. This shows that considering languages as cost functions does not lose any information on these languages, and therefore cost function theory properly extends language theory.

\begin{rem}
They are uncountably many cost functions in $\mathcal{F}/\approx$, and each cost function contains uncountably many functions. Therefore it is hard to give an explicit description of all the functions in a $\approx$-class, other than all the functions equivalent to a particular representative.
\end{rem}

We will now introduce two models of cost automata recognizing cost functions. These definitions are from \cite{Col09}, the reader can report to it for more details. In both cases, we define the semantic of an automaton $\aut$ as a function $\sem{\aut}$ in $\mathcal F$, which we will mainly look as a representative of the cost function $\sem{\aut}^\approx$.

\subsection{$B$-automata}\label{sec:aut}
A $B$-automaton is a tuple $\perm{Q,\A,\mathit{In},\mathit{Fin},\Gamma,\Delta}$ where $Q$ is the set of states, $\A$ the alphabet, $\mathit{In}$ and $\mathit{Fin}$ the sets of initial and final states, $\Gamma$ the set of counters, and $\Delta\subseteq Q\times\A\times\set{\e,\ic,\r}^\Gamma\times Q$ is the set of transitions.

Counters have integers values starting at $0$, and an atomic action $\sigma\in\set{\e,\ic,\r}^\Gamma$ update the value of every counter $\gamma$ in the following way: $\ic$ increments by $1$, $\r$ resets to $0$, and $\e$ leaves the counter value unchanged. If $e$ is a run, let $C(e)$ be the set of values reached during $e$, at any point of the run and on any counter of $\Gamma$. The notation ``$\ic$'' stands for ``increment check'', meaning that as soon as we increment a counter, we put its value in $C(e)$.

A $B$-automaton $\aut$ recognizes a cost function $\semB{\aut}^\approx$ via the following semantic: 
$$\semB{\aut}(u)=\inf\set{\sup C(e), e\text{ run of }\aut\text{ over }u}.$$

With the usual conventions that $\sup\emptyset=0$ and $\inf\emptyset=\infty$.
It means that the value of a run is the maximal value reached by a counter, and the nondeterminism resolves in taking the run with the least value.
If there is no accepting run on a word $u$, then $\semB{\aut}(u)=\infty$.

Notice that in particular, if the automaton does not have any counter, then it is a classical automaton recognizing a language $L$, and its semantic is $\semB{\aut}=\chi_L$, with $\chi_L(u)=0$ if $u\in L$ and $\chi_L(u)=\infty$ if $u\notin L$.

\begin{exa}\label{example:automata}
Let $\A=\set{a,b}$. The functions $|\cdot|_a$ and $2|\cdot|_a+5$ represent the same cost function, which is recognized by the following one-counter $B$-automaton on the left-hand side. The cost function containing $u\mapsto\min\set{n\in\N,~a^n\text{ factor of }u}$ is recognized by the nondeterministic one-counter $B$-automaton on the right-hand side.
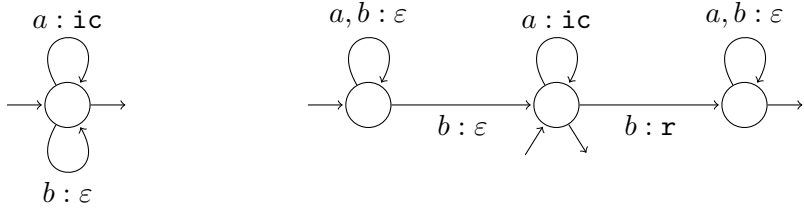
\begin{center}
\begin{tikzpicture}[shorten >=1pt,node distance=2.5cm,on grid,auto,initial text=,
every state/.style={inner sep=0pt,minimum size=6mm}, accepting/.style=accepting by arrow]
   	\node[state,initial,accepting] 	(q_0) {};

   \path[->] 
    (q_0) 	edge [in=60,out=120,loop] node {$a:\ic$} ()
    		edge [in=-60,out=-120,loop] node [anchor=north] {$b:\e$} ();

\begin{scope}[xshift=4cm]
  \node[state,initial]	(q_1) {};    		
  \node[state] (q_2)[right=of q_1] {};
		\node at (2,-.8) (i) {};
    \node at (3,-.8) (f) {};
   \node[state,accepting] (q_3) [right=of q_2] {};
    
  \path[->] 
    (q_1) edge [in=60,out=120,loop] node  {$a,b:\e$} ()
    		edge [right] node [below] {$b:\e$} (q_2)
		(i) edge (q_2)
    (q_2) 	edge [in=60,out=120,loop] node {$a:\ic$ } ()
    		edge [right] node [below] {$b:\r$} (q_3)
				edge (f)
    (q_3)  edge [in=60,out=120,loop] node  {$a,b:\e$} ();
    \end{scope}		
\end{tikzpicture}

\end{center}
\end{exa}

\subsection{$S$-automata}\label{sec:Saut}
The model of $S$-automaton is dual to the one $B$-automaton. The aim of this model is to mimic completation: as we cannot complement a function, we get around it by reversing the semantic of the automata defining it.

An $S$-automaton is a tuple $\perm{Q,\A,\mathit{In},\mathit{Fin},\Gamma,\Delta}$ where $Q$ is the set of states, $\A$ the alphabet, $\mathit{In}$ and $\mathit{Fin}$ the sets of initial and final states, $\Gamma$ the set of counters, and $\Delta\subseteq Q\times\A\times\set{\e,\i,\r,\ccr}^\Gamma\times Q$ is the set of transitions.

Counters have integers values starting at $0$, and an action $\sigma\in(\set{\e,\i,\r,\ccr}^*)^\Gamma$ performs a sequence of atomic actions on each counter, where atomic actions are either $\i$ (increment by $1$), $\r$ (reset to $0$), $\e$ (do nothing on the counter), or $\ccr$ (check the counter value and reset it). If $e$ is a run, let $C(e)$ be the set of values checked during $e$ on all counters of $\Gamma$. This means that this time, contrary to what happened in $B$-automata, we only put in $C(e)$ values witnessed during an operation $\ccr$. This is because we will be interested in the minimum of these values, and therefore we do not want to observe all intermediate values.

An $S$-automaton $\aut$ computes a cost function $\semS{\aut}^\approx$ via the following semantic : 
$$\semS{\aut}(u)=\sup\set{\inf C(e), e\text{ run of }\aut\text{ over }u}.$$

Notice that $\inf$ and $\sup$ have been switched, compared to the definition of the $B$-semantic. It means that the value of a run of an $S$-automaton is the minimal checked value, and the automaton tries to maximize its value among all runs.

In particular, if $\aut$ is a classical automaton for $L$, then its $S$-semantic is $\semS{\aut}=\chi_{\notL}$, where $\notL$ is the complement of $L$. This conforts the intuition that switching between $B$ and $S$-automata corresponds to complementation.

\begin{exa}\label{exaSaut}
We will redefine the two cost functions from example \ref{example:automata}, this time with $S$-automata. The first one counts the number of $a$, and guess the last letter to check the value. Notice that the exact function it computes is between $|\cdot|_a-1$ and $|\cdot|_a$, so is equivalent to $|\cdot|_a$ up to $\approx_\alpha$, with $\alpha(x)=x+1$.
The second automaton counts all blocks of $a$, and also needs guess the last letter, in order to count the last block (-1 if the last letter is $a$).
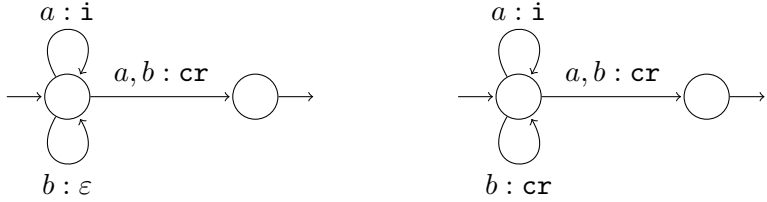
\begin{center}
\begin{tikzpicture}[shorten >=1pt,node distance=2.5cm,on grid,auto,initial text=,
every state/.style={inner sep=0pt,minimum size=6mm}, accepting/.style=accepting by arrow]
   	\node[state,initial] 	(q_0) {};
		\node[state,accepting,right=of q_0] 	(q_f) {};
   \path[->] 
    (q_0) 	edge [in=60,out=120,loop] node {$a:\i$} ()
    		edge [in=-60,out=-120,loop] node [anchor=north] {$b:\e$} ()
				edge node{$a,b:\ccr$} (q_f);

\begin{scope}[xshift=6cm]
    	\node[state,initial] 	(q_1) {};
		\node[state,accepting,right=of q_1] 	(q_2) {};
 
   \path[->] 
    (q_1) 	edge [in=60,out=120,loop] node {$a:\i$} ()
    		edge [in=-60,out=-120,loop] node [anchor=north] {$b:\ccr$} ()
				edge node{$a,b:\ccr$} (q_2);
    \end{scope}		
\end{tikzpicture}

\end{center}

\end{exa}

\begin{thm}\cite{Col09}\label{BSeq}
If $f^\approx$ is a cost function, there is a $B$-automaton for $f^\approx$ if and only if there is an $S$-automaton for $f^\approx$.  That is to say, $B$ and $S$-automata have same expressive power (up to $\approx$) in term of recognized cost functions.
\end{thm}
%
%Therefore, we will say that a cost function $f^\approx$ is \intro{regular} if it is recognized either by a $B$-automaton or a $S$-automaton (or equivalently, by both formalisms).

\section{Stabilization semigroups}

\subsection{Classical ordered semigroups, and regular languages}

An ordered semigroup is a tuple $\semi=\perm{S,\cdot,\leq}$, where $\cdot$ is a product $S\times S\to S$, and $\leq$ is a partial order compatible with $\cdot$, i.e. $\forall x,y,z \in S, x\leq y$ implies $z\cdot x\leq z\cdot y$ and $x\cdot z\leq y\cdot z$. We will always write $\semi$ for the whole structure, and $S$ for the underlying set.

An \intro{ideal} of $\semi$ is a set $I\subseteq S$ which is $\leq$-closed, i.e. such that for all $x\in I$ and $y\leq x$, we have $y\in I$.

We remind how a classical semigroup can recognize a regular language $L\subseteq \A^*$. The order is not necessary here.
Let $h:\A\to S$ be a function, canonically extended to a morphism $h:\A^+\to S$.
Let $P\subseteq S$ be a subset of $S$, called \intro{accepting subset}.

Then the language recognized by $\semi,h,P$ is $L=h^{-1}(P)$. It is well-known that a language is regular if and only if it can be recognized by a finite semigroup.

This section explains how to generalize this to the cost functions setting, as it was done in \cite{Col09}.

\subsection{Cost sequences}\label{subsection:suites}
The aim is to give a semantic to stabilization semigroups. Some mathematical preliminaries are required.

Let~$(E,\leq)$ be an ordered set, $\alpha$ a function from $\N$ to $\N$,
and~$\vec a,\vec b\in E^\N$ two infinite sequences. We define the relation $\preceq_\alpha$ by $\vec a$\intro{$\preceq_\alpha$}$\vec b$ if :
$$
\forall n.\forall m.\quad \alpha(n)\leq m\rightarrow \vec a(n)\leq \vec b(m)\ .
$$
A sequence $\vec a$ is said to be \intro{$\alpha$-non-decreasing} if~$\vec a\preceq_\alpha\vec a$.
We define~\intro{$\sim_\alpha$} as~$\preceq_\alpha\cap\succeq_\alpha$,
and $\vec a$\intro{$\preceq$}$\vec b$ (resp. $\vec a$\intro{$\sim$}$\vec b$) if $\vec a\preceq_\alpha\vec b$
(resp. $\vec a \sim_\alpha\vec b$) for some~$\alpha$.

\emph{Remarks:}
\begin{itemize}
\item if~$\alpha\leq\alpha'$ 
then~$\vec a\preceq_\alpha\vec b$ implies $\vec a\preceq_{\alpha'}\vec b$,
\item if~$\vec a$ is~$\alpha$-non-decreasing, then it is~$\alpha$-equivalent
to a non-decreasing sequence,
\item $\vec a$ is $\mathit{id}$-non-decreasing iff it is non-decreasing,
\item let~$\vec a,\vec b\in E^\N$ be two non-decreasing sequences, then
$\vec a\preceq_\alpha\vec b$ iff~$\vec a\circ\alpha\leq\vec b$.
\end{itemize}
The $\alpha$-non-decreasing sequences ordered by~$\preceq_\alpha$ can be seen as a weakening of the~$\alpha=\mathit{id}$ case.
We will identify the elements~$a\in E$ with the constant sequence of value~$a$.

The relations~$\preceq_\alpha$ and~$\sim_\alpha$ are not transitive, but the following property guarantees a certain kind of transitivity.
\begin{fact}\label{fact:transcost}
$\vec a\preceq_\alpha \vec b\preceq_\alpha\vec c$ implies $\vec a\preceq_{\alpha\circ\alpha}\vec c$ and
$\vec a\sim_\alpha \vec b\sim_\alpha \vec c$ implies $\vec a\sim_{\alpha\circ\alpha}\vec c$.
\end{fact}

The function~$\alpha$ is used as a ``precision'' parameter for~$\sim$ and~$\preceq$.
Fact \ref{fact:transcost} shows that a transitivity step costs some precision.
%Le but est d'obtenir une extension de~$E$ au moyen de suites $\alpha$-croissantes, avec comme point de départ l'identification d'un élément de $E$ avec la suite constante égale à cet élément.
For any~$\alpha$, the relation~$\preceq_\alpha$ coincides over constant sequences with order $\leq$ (up to identification of constant sequences with their constant value). Consequently, the infinite sequences in $E^\N$ ordered by $\preceq_\alpha$ form an extension of~$(E,\leq)$.

In the following, while using relations $\preceq_\alpha$ and~$\sim_\alpha$, we may forget the subscript $\alpha$ and verify instead that the proof has a bounded number of transitivity steps.

\begin{defi}
Let $\perm{S,\cdot,\leq}$ be an ordered semigroup and $I$ be an ideal of $S$.
\begin{itemize}
\item If $\vec a\in S^\N$ is an $\alpha$-non-decreasing sequence of elements of $S$, we note $$I[\vec a]=\inf\set{n\in\N : \vec a(n)\notin I}.$$ In other words, $I[\vec a]$ is the first position where $\vec a$ gets out of $I$.
\item If $x,y\in S$ and $m\in\N$, we define the cost sequence $x|_m y$ by $(x|_m y)(n)=\begin{cases} x\text{ if }n\leq m\\ y\text{ otherwise}\end{cases}$.
\end{itemize}
\end{defi}

\subsection{Stabilization semigroups}
The notion of stabilization semigroup is introduced in \cite{Col09}, in order to extend the classic notion of semigroups, and recognize cost functions instead of languages. If $\semi=\perm{S,\cdot}$ is a semigroup (possibly with other operations), we will note $E(\semi)$ the set of idempotent elements of $\semi$, i.e. elements $e\in S$ such that $e\cdot e=e$.

\begin{defi}
A  \intro{stabilization semigroup} $\semi=\perm{ S,\cdot,\leq,\sharp}$
is an ordered semigroup $\perm{S,\cdot,\leq}$
together with an operator~\intro{$\sharp$}$:E(\semi)\rightarrow E(\semi)$ (called \intro{stabilization}) such that:
\begin{itemize}
\item for all~$a,b\in S$with~$a\cdot b\in E(\semi)$ and $b\cdot a\in E(\semi)$, $(a\cdot b)^\sharp=a\cdot(b\cdot a)^\sharp\cdot b$;
\item for all~$e\in E(\semi)$, $(e^\sharp)^\sharp=e^\sharp\leq e$;
\item for all~$e\leq f$ in~$E(\semi)$, $e^\sharp\leq f^\sharp$;
\item if $\semi$ is a monoid, $1^\sharp=1$, we say then that $\semi$ is a \intro{stabilization monoid}
\end{itemize}
\end{defi}

In this paper, we only consider finite stabilization semigroups.
The intuition of the $\sharp$ operator is that $e^\sharp$ means "$e$ repeated many times",
which appears in the following properties, consequences of the definition above : 
$$
e^\sharp=e\cdot e^\sharp=e^\sharp\cdot e=e^\sharp\cdot e^\sharp=(e^\sharp)^\sharp\leq e
$$

\subsection{Factorization trees and compatible function}
Let $\semi=\perm{S,\cdot,\leq,\sharp}$ be a stabilization semigroup, and $u\in S^*$.
A \intro{$n$-tree} $t$ over $u$ is a $S$-labelled tree such that $u$ is the leaf word of $t$, and for each node $p$ of $t$, we are in one of these case :

\begin{description}
\item[Leaf] $p$ is a leaf,
\item[Binary : ] $p$ has only $2$ children $p_1,p_2$, and $t(p)=t(p_1)\cdot t(p_2)$,
\item[Idempotent : ] $p$ has $k$ children $p_1,\dots,p_k$ with $k\leq n$, and there is $e\in E(\semi)$ such that $t(p)=t(p_1)=\dots=t(p_k)=e$,
\item[Stabilization : ] $p$ has $k$ children $p_1,\dots,p_k$ with $k> n$, and there is $e\in E(\semi)$ such that $t(p_1)=\dots=t(p_k)=e$, and $t(p)=e^\sharp$.
\end{description}

The root of $t$ is called its \intro{value} and is noted $\val(t)$.

\begin{exa}
Let $u=abaaabbbbaaabbba$, $n\in\N$,  and $S=\set{a,b,\bot}$ with $aa=ab=a$, $bb=b^\sharp=b$, and $a^\sharp=\bot$. The following tree is an $n$-tree over $u$ :

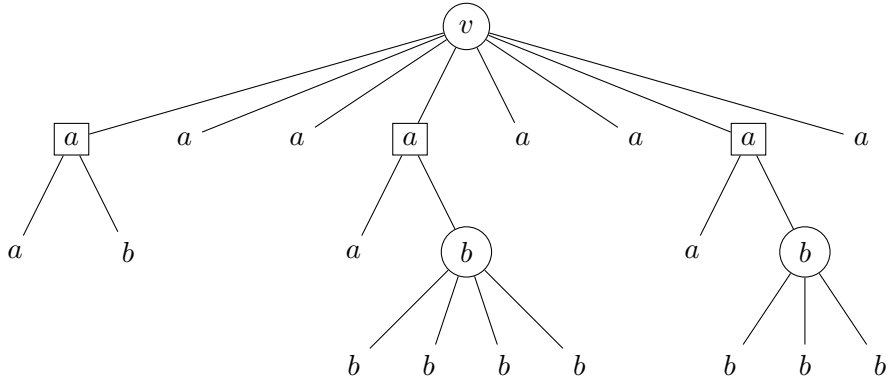
\begin{center}
\begin{tikzpicture}[level 3/.style={sibling distance=1cm}]
         \node[draw,circle] {$v$}
                child {node[draw] {$a$}
                	child {node[] {$a$}}
                	child {node[] {$b$}}}
                child {node[] {$a$}}
                child {node[] {$a$}}
                child {node[draw] {$a$}
                    child {node[] {$a$}}
             		child {node[draw,circle] {$b$} 
             			child {node[] {$b$}}
             			child {node[] {$b$}} 
             			child {node[] {$b$}}
             			child {node[] {$b$}}
             			}
             		}
             	child {node[] {$a$}}
				child {node[] {$a$}}             	
             	child {node[draw] {$a$}
                    child {node[] {$a$}}
             		child {node[draw,circle] {$b$} child {node[] {$b$}}  child {node[] {$b$}} child {node[] {$b$}}}}
           	    child {node[] {$a$}}
                 ;

\end{tikzpicture}
\end{center}
Notice that the number of children of the root is $|u|_a=8$.
Two cases are possible :
\begin{itemize}
\item $n\leq 8$ : the root is an idempotent node, and $v=a$.
\item $n> 8$ : the root is a stabilisation node, and $v=a^\sharp=\bot$.
\end{itemize}
This gives an intuition of how these factorization trees can be used to associate a value to a word, here its number of occurences of $a$. 
\end{exa}

In the following we will establish formally how we can use factorization trees to give a semantic to stabilization semigroups.

The following theorem is the cornerstone of this process. This theorem is a deep combinatoric result and generalizes Simon's factorization forests theorem. It can be considered as a Ramsey-like theorem, because it provides the existence of big well-behaved structures (the factorization tree, and in particular the idempotent nodes) if the input word is big enough.

\begin{thm}\cite{Col09}\label{height}
For all $\semi=\perm{S,\cdot,\leq,\sharp}$, there exists $H\in\N$ such that for all $u\in S^*$ and $n\in\N$, there is a $n$-tree over $u$ of height at most $H$.
\end{thm}

This allows us to define $\rho:S^+\to\N \to S$ by $$\rho(u)(n)=\min\set{\val(t) : t\text{ is an $n$-tree over $u$ of height at most $H$}}.$$
The function $\rho$ is called \intro{compatible} with $\semi$. It depends on $H$ so there may be several compatible functions, however we will see that they are equivalent in some sense.

If $\rho$ is a function $S^+\to\N \to S$, we associate to it a function $\tilde\rho:((S^+)^{\N})\to\N\to S$ by 
$\tilde\rho(\vec u)(n):=\rho({\vec u}(n))(n)$. We will also identify elements of $(S^{\N})^+$ with their canonic image in $(S^+)^{\N}$ (i.e. view a word of sequences as a sequence of words of same length).

\begin{thm}\cite{Col09}\label{compaxioms}
If $\rho$ is a compatible function of~$\semi$, then there exists~$\alpha$ such that :
\begin{description}
%\item[Monotonicity.] $\rho$ is~$\alpha$-monotone,
\item[Letter.] for all $a\in S,n\in\N$, $\rho(a)(n)=a$,
\item[Product.] for all $a,b\in S$, $\rho(ab)\sim_\alpha a\cdot b$,
\item[Stabilization.] for all~$e\in E(\semi)$, $m\in\N$, 
	$\rho(e^m)\sim_\alpha(e^\sharp|_me)$,
\item[Substitution.] for all~$u_1,\dots, u_n\in S^+$, $n\in\N$, $\rho(u_1\dots u_n)\sim_\alpha\tilde\rho(\rho(u_1)\dots\rho(u_n))$ (we identify sequence of words and word of sequences)
\end{description}
\end{thm}

\begin{exa}\label{example:compatible-ex}
Let~$\semi$ be the stabilization semigroup with $3$ elements $\bot\leq a\leq b$,
with product defined by :~$x\cdot y=\min_\leq(x,y)$ ($b$ neutral element),
and stabilization by~$b^\sharp=b$ and~$a^\sharp=\bot^\sharp=\bot$.
Let~$u\in\{\bot,a,b\}^+$, we define $\rho$ by:
$$
\rho(u)=\begin{cases}
			b&\text{if}~u\in b^+\\
			\bot|{|u|_a}a&\text{if}~u\in b^*(ab^*)^+\\
			\bot&\text{otherwise.}
			\end{cases}
$$
Then~$\rho$ is compatible with~$\semi$.
This is proved by building a factorization tree of height $3$, with idempotent (or stabilisation) $b$-nodes at level $3$, binary nodes $a=a\cdot b$ at level $2$, and one idempotent/stabilisation node at level $1$.
\end{exa}

%The following lemma links $\rho$ with the classical product $\pi:S^+\to S$.
%
%\begin{lem}\label{lemma:rho-pi}\cite{CKL}
%Let $\rho$ compatible with a semigroup $\semi$.
%There exists $\gamma$ such that for any $n\in\N$ and $u\in S^+$, if $|u|\leq n$ then for all $k\geq \gamma(n), \rho(u)(k)=\pi(u)$
%\end{lem}
%\begin{proof}
%We show this result by induction on $n$. It is true for $n=1$ by taking $\gamma(1)=1$.
%We assume $\gamma(k)$ constructed for $k<n$, and we want to show the result for $n$.
%Let $u\in S^+$ of length $n$, $u=va$ with $|v|=n-1$ and $a\in S$.
%Let$\alpha$ a witness of $\rho$ compatible with $\semi$. The substitution property tells us that $\rho(u)\sim_\alpha \tilde\rho(\rho(v)a)$. but by induction hypothesis, for all $k\geq\gamma(n-1)$, 
%$\tilde\rho(\rho(v)a)(k)=\rho(\rho(v)(k)a)(k)=\rho(\pi(v)a)(k)$. Moreover, $\rho(\pi(v)a)\sim_\alpha \pi(v)\cdot a=\pi(u)$.
%Hence we have for all $k\geq\alpha(\gamma(\alpha(n-1))), \rho(u)(k)=\pi(u)$.We get the result with $\gamma(n)=\alpha(\gamma(\alpha(n-1)))$.
%\end{proof}

\subsection{Recognized cost functions}
We now have all the mathematical tools to define how stabilization semigroups can recognize cost functions.

Let $\semi=\langle S,\cdot,\leq,\sharp\rangle$ be a stabilization semigroup.
Let $h:\A\rightarrow S$ be a morphism, canonically extended to $h:\A^+ \rightarrow S^+$,
and~$I\subseteq S$ an ideal. let $\rho$ be a compatible function associated with $\semi,h$.
We say that the quadruple $\semi,h,I,\rho$ \intro{recognizes} the function $f:\A^+\rightarrow\Ninfty$ defined by 
$$f(u)=I[\rho(h(u))]= \inf\set{n\in\N : \rho(h(u))(n)\notin I}.$$

We say that $I$ is the \intro{accepting ideal} of $\semi$, it generalizes the accepting subset $P$ used in the classical setting.

Indeed, if $\semi,h,P$ is a classical semigroup recognizing $L\subseteq\A^+$ with an accepting subset $P$, we can take $\rho$ to be the normal product $\pi:S^+\to S$, $\sharp$ to be the identify function on idempotents, and $I$ to be the complement of $P$. Then $\semi,h,I,\pi$ recognizes $\chi_L$.

\begin{thm}\label{unirho}\cite{Col09}
If $\rho'$ satifies all the properties given in Theorem \ref{compaxioms}, then $\rho'\sim\rho$. In other words, $\rho$ is unique up to $\sim$ (and in particular the choice of $H$ is not important).
Moreover, if $\semi,h,I$ is given, and $\rho\sim\rho'$ are two compatible functions for $\semi$, then the functions defined by $\semi,h,I$ relatively to $\rho$ and $\rho'$ are equivalent up to $\approx$. This allows us to uniquely define the cost function $F=f^\approx$ recognized by the triplet $\semi,h,I$, without ambiguity.
\end{thm}

%By Proposition \ref{prop:ideal-cost}, the recognized cost function does not depend on the choice of $\rho$.
% On dira que $f:alphabet^*\rightarrow\Ninfty$ est reconnue par $\semi,h,I$ si sa restriction à $\A^+$ est reconnue par $\semi,h,I$.

\begin{exa}\label{recognizable-ex}
Let~$\A=\{a,b\}$, the cost function $|\cdot|_a^\approx$ is recognizable.
We take the stabilization semigroup from Example \ref{example:compatible-ex},
$h$ defined by~$h(a)=a,h(b)=b$, and~$I=\{\bot\}$.
We have then $|u|_a=\inf\set{n\in\N : \rho(h(u))(n)\neq\bot}$ for all~$u\in \A^+$. 
\end{exa} 

The following theorem links cost automata with stabilization semigroups, and allows us to define the class of regular cost functions.

\begin{thm}\cite{Col09}
Let $F$ be a cost function, the following assertions are equivalent:
\begin{itemize}
\item $F$ is recognized by a $B$-automaton,
\item $F$ is recognized by an $S$-automaton,
\item $F$ is recognized by a finite stabilization semigroup.
\end{itemize}
Such a cost function $F$ will be called \intro{regular} by generalization of this notion from language theory.
\end{thm}

Notice that if $L\subseteq\A^+$ is a language, then $\chi_L^\approx$ is a regular cost function if and only if $L$ is a regular language. This shows that the notion of regularity for cost function is a proper extension of the one from language theory. That is to say, restricting cost functions theory to [regular] cost functions of the form $\chi_L^\approx$, one exactly gets [regular] language theory.

\section{Quantitative LTL}\label{ltl}

We will now use an extension of LTL to describe some regular cost functions. This has been done successfully with regular languages, so we aim to obtain the same kind of results. Can we still go efficiently from an LTL-formula to an automaton? 

\subsection{Definition}
The first thing to do is to extend LTL so that it can decribe cost functions instead of languages. We must add quantitative features, and this will be done by a new operator $\UN$, required to appear positively in the formula. Unlike in most uses of LTL, we work here over finite words. This is to avoid additional technical considerations due to new formalisms suited to infinite words, which would make all the proofs heavier without adding any new ideas. 

Formulas of $\ltlq$ (on finite words on an alphabet $\A$) are defined by the following grammar : 
$$\varphi:=a~|~\varphi\wedge\varphi~|~\varphi\vee\varphi~|~X\varphi~|~\varphi U\varphi~|~\varphi\UN\varphi~|~\Omega$$

Where $N$ is a unique free variable, common for all occurences of $\UN$ operator. This is in the same spirit as in \cite{PromptLTL}, where the bound is global for all the formula.

\begin{itemize}
\item $a$ means that the current letter is $a$, $\wedge$ and $\vee$ are the classical conjunction and disjunction;
\item $X\varphi$ means that $\varphi$ is true at the next letter;
\item $\varphi U\psi$  means that $\psi$ is true somewhere in the future, and $\varphi$ holds until that point;
\item $\varphi\UN\psi$ means that $\psi$ is true somewhere in the future, and $\varphi$ can be false at most $N$ times before $\psi$. 
\item $\Omega$ means that we are at the end of the word.
\end{itemize}

Notice the absence of negation in the syntax of $\ltlq$. However, we can still consider that $\ltlq$ generalizes classical LTL (with negation), because an LTL formula can be turned into an $\ltlq$ formula by pushing negations to the leaves. That is why we heed operators in dual forms in the syntax.
Remark that we do not need a dual operator for $U$, because we can use $\Omega$ to negate it: $\neg (\varphi U\psi)\equiv \neg\psi U (\neg\varphi\vee\Omega)$.
Moreover we can also express negations of atomic letters: for all $a\in\A$ we can define $\neg a=(\bigvee_{b\neq a} b)\vee\Omega$ to signify that the current letter is not $a$. 

We can then choose any particular $a\in\A$, and define $\top=a\vee\neg a$ and $\bot=a\wedge\neg a$, meaning respectively true and false.

We also define connectors ``eventually'' : $F\varphi=\top U \varphi$ and ``globally'' : $G\varphi=\varphi U \Omega$.

%\begin{rem}
%The negation of an $\ltlq$-formula is not in general an $\ltlq$-formula.
%\end{rem}

\subsection{Semantics}
We want to associate a function $\sem{\varphi}$ to any $\ltlq$-formula $\varphi$. As usual, we will often be more interested in the cost function $\sem{\varphi}^\approx$ recognized by $\varphi$.

We will say that $(u,n)\models\varphi$ ($u,n$ is a model of $\varphi$) if $\varphi$ is true on $u$ with $n$ as valuation for $N$, i.e. as number of errors for all the $\UN$'s in the formula $\varphi$. We finally define
$$\sem{\varphi}(u)=\inf\set{n\in\N/ (u,n)\models\varphi}$$

We can remark that if $(u,n)\models\varphi$, then for all $k\geq n, (u,k)\models\varphi$, since the $\UN$ operators appear always positively in the formula.
% In particular, $\sem{\varphi}(u)=0$ means that $\forall n\in\N,(u,n)\models\varphi$, and $\sem{\varphi}(u)=\infty$ means that $\forall n\in\N,(u,n)\not\models\varphi$ (since $\inf\emptyset=\infty$).

%In order to define a meaningful semantic, we don't allow negative occurences of $\UN$. That is why the negation of an $\ltlq$-formula is not in general an $\ltlq$-formula.

\begin{prop}~
\begin{itemize}
\item $\sem{a}(u)=0$ if $u\in a\A^*$, and $\infty$ otherwise
\item $\sem\Omega(u)=0$ if $u=\varepsilon$, and $\infty$ otherwise
\item $\sem{\varphi\wedge\psi}=\max(\sem\varphi,\sem\psi)$, and $\sem{\varphi\vee\psi}=\min(\sem\varphi,\sem\psi)$
\item $\sem{X\varphi}(au)=\sem\varphi(u)$, $\sem{X\varphi}(\varepsilon)=\infty$
\item $\sem\top=0$, and $\sem\bot=\infty$
\end{itemize}
\end{prop}

\begin{exa}
Let $\varphi=(\neg a)\UN\Omega$, then $\sem{\varphi}=|\cdot|_a$
\end{exa}

We use $\ltlq$-formulae in order to describe cost functions, so we will often work modulo cost function equivalence $\approx$. However, we will sometimes be interested in the exact function $\sem{\varphi}$ described by $\varphi$.

\begin{rem}
If $\varphi$ does not contain any operator $\UN$, $\varphi$ is a classical LTL-formula computing a language $L$, and $\sem{\varphi}=\chi_L$.
\end{rem}

\section{From $\ltlq$ to $B$-Automata}\label{sec:ltlBaut}
\newcommand{\sub}{\mathrm{sub}}

\subsection{Description of the automaton}\label{descB}

We will now give a direct translation from $\ltlq$-formulae to $B$-automata, i.e. given an $\ltlq$-formula $\phi$ on a finite alphabet $\A$, we want to build a $B$-automaton recognizing $\sph^\approx$. We will also show that a slight change in the model of $B$-automaton (namely allowing sequences of counter actions on transitions) allows us to design a $B$-automaton $\aut_\phi$ with $\sem{\aut_\phi}=\sem{\phi}$: the functions recognized by $\aut_\phi$ and $\phi$ are equal and not just equivalent up to $\approx$.
This construction is adapted from the classic translation from LTL-formula to B\"uchi automata \cite{DG}.%citer Vardi 86
\bigskip

\newcommand{\YF}{\emptyset}

Let $\phi$ be an $\ltlq$-formula.
We define $\sub(\phi)$ to be the set of subformulae of $\phi$, and $Q=2^{\sub(\phi)}$ to be the set of subsets of $\sub(\phi)$.

We want to define a $B$-automaton $\aut_{\phi}=\perm{Q,\A,\mathit{In},\mathit{Fin},\Gamma,\Delta}$ such that $\semB{\aut}\approx\sph$.

We set the initial states to be $\mathit{In}=\set{\set{\phi}}$ and the final ones to be $\mathit{Fin}=\set{\YF,\set{\Omega}}$
We choose as set of counters $\Gamma=\set{\gamma_1,\dots,\gamma_k}$ where $k$ is the number of occurences of the $\UN$ operators in $\phi$, labeled from $\UN_1$ to $\UN_k$.

A state is basically the set of constraints we have to verify before the end of the word, so the only two accepting states are the one with no constraint, or with only constraint to be at the end of the word.
% For every letter $a\in\A$, there is a transition $\YF\overset{a : \varepsilon}{\longrightarrow}\YF$. However, since we are here in the case of finite words, there are no outgoing transitions from states containing formula $\Omega$.

The following definitions are the same as for the classical case (LTL to B\"uchi automata): 

\newcommand{\nex}{\mathrm{next}}

\begin{defi}
~
\begin{itemize}
\item An atomic formula is either a letter $a\in\A$ or $\Omega$
\item A set $Z$ of formulae is consistent if there is at most one atomic formula in it.
\item A reduced formula is either an atomic formula or a Next formula (of the form $X\varphi$).
\item A set $Z$ is reduced if all its elements are reduced formulae.
\item If $Z$ is consistent and reduced, we define $\nex(Z)=\set{\varphi/ X\varphi\in Z}$.
\end{itemize}
\end{defi}

\begin{lem}[Next Step]
If $Z$ is consistent and reduced, for all $u\in\A^*, a\in\A$ and $n\in\N$,

$$(au,n)\models \bigwedge Z \text{  iff  } (u,n)\models\bigwedge\nex(Z) \text{ and }Z\cup\set{a}\text{ consistent }$$

\end{lem}
\begin{proof}
If $(au,n)\models \bigwedge Z$, then the only atomic formula that $Z$ can contain is $a$, and therefore $Z\cup\set{a}$.
Moreover, for every formula of the form $X\varphi$ in $Z$, we have $(au,n)\models X\varphi$. By definition of the semantic of the $X$ operator, this means $(u,n)\models\varphi$. This is true for every $X\varphi$ in $Z$, so we obtain $(u,n)\models\nex(Z)$. The converse is similar.
\end{proof}

We would like to define $\aut_\phi$ with $Z\longrightarrow\nex(Z)$ as transitions.
% where $Z\cup\set{a}$ is consistent.
%Remark that $\emptyset\overset{a : \varepsilon}{\longrightarrow}\emptyset$ for all $a\in\A$.

The problem is that $\nex(Z)$ is not consistent and reduced in general. If $\nex(Z)$ is inconsistent we remove it from the automaton. If it is consistent, we need to apply some reduction rules to get a reduced set of formulae. This consists in adding $\varepsilon$-transitions (but with possible actions on the counter) towards intermediate sets which are not actual states of the automaton (we will call them "pseudo-states"), until we reach a reduced set.

Let $\psi$ be maximal (in size) not reduced in $Y$, we add the following transitions
\begin{itemize}
\item If $\psi=\varphi_1\wedge\varphi_2$ :   $Y\overset{\varepsilon : \varepsilon}{\longrightarrow}Y\setminus\set{\psi}\cup\set{\varphi_1,\varphi_2}$

\item If $\psi=\varphi_1\vee\varphi_2$ : 
$\left\{\begin{array}{l}
Y\overset{\varepsilon : \varepsilon}{\longrightarrow}Y\setminus\set{\psi}\cup\set{\varphi_1}\\
Y\overset{\varepsilon : \varepsilon}{\longrightarrow}Y\setminus\set{\psi}\cup\set{\varphi_2}
\end{array}\right.$ 

\item If $\psi=\varphi_1 U\varphi_2$ :
$\left\{\begin{array}{l}
Y\overset{\varepsilon : \varepsilon}{\longrightarrow}Y\setminus\set{\psi}\cup\set{\varphi_1,X\psi}\\
Y\overset{\varepsilon : \varepsilon}{\longrightarrow}Y\setminus\set{\psi}\cup\set{\varphi_2}\\
\end{array}\right.$

\item If $\psi=\varphi_1\UN_j\varphi_2$ :
$\left\{\begin{array}{l}
Y\overset{\varepsilon : \varepsilon}{\longrightarrow}Y\setminus\set{\psi}\cup\set{\varphi_1,X\psi}\\
Y\overset{\varepsilon : ic_j}{\longrightarrow}Y\setminus\set{\psi}\cup\set{X\psi}\text{ (we count one mistake)}\\
Y\overset{\varepsilon : r_j}{\longrightarrow}Y\setminus\set{\psi}\cup\set{\varphi_2}\\
\end{array}\right.$

where action $r_j$ (resp. $ic_j$) perform $r$ (resp. $ic$) on counter $\gamma_j$ and $\varepsilon$ on the other counters.
\end{itemize}
The pseudo-states do not (a priori) belong to $Q=2^{\sub(\phi)}$ because we add formulae $X\psi$ for $\psi\in\sub(\phi)$, so if $Z$ is a reduced pseudo-state, $\nex(Z)$ will be in $Q$ again since we remove the new next operators.
\smallskip

%Let us observe that it is not wrong to have only one counter shared for all duplications of formulae $\varphi_1\UN_j\varphi_2$ occuring in the different states and pseudo-states, since if $\varphi_2$ is satisfied, we can stop counting mistakes for all waiting instances of $\varphi_1\UN_j\varphi_2$, only the oldest one will matter.

The transitions of automaton $\aut_\phi$ will be defined as follows:
$$\Delta=\set{Y\overset{a:\sigma}{\longrightarrow}\nex(Z)~|~ Y\in Q, Z\cup\set{a}\text{ consistent and reduced}, Y\overset{\varepsilon:\sigma}{\longrightarrow}_*Z}$$
%Let $\red(Y)=\set{Z\text{ consistent and reduced}/ Y\overset{\varepsilon}{\longrightarrow}^*Z}$
where $Y\overset{\varepsilon:\sigma}{\longrightarrow}_*Z$ means that there is a sequence of $\varepsilon$-transitions from $Y$ to $Z$ with $\sigma$ as combined action on counters.

\subsection{Correctness of $\aut_\phi$}\label{corrB}

We will now prove that $\aut_\phi$ is correct, i.e. computes the same cost function as $\phi$.

\begin{defi}
If $\sigma$ is a sequence of actions on counters, we will call $\val(\sigma)$ the maximal value checked on a counter during $\sigma$ with $0$ as starting value of the counters, and $\val(\sigma)=0$ if there is no check in $\sigma$. It corresponds to the value of a run of a $B$-automaton with $\sigma$ as combined action of the counter.
\end{defi}

\begin{lem}\label{correct}
Let $u=a_1\dots a_m$ be a word on $\A$ and $Y_0\overset{a_1 : \sigma_1}{\rightarrow}Y_1\overset{a_2 : \sigma_2}{\rightarrow}\dots\overset{a_m : \sigma_m}{\rightarrow}Y_m$ an accepting run of $\aut_\phi$.

Then for all $\psi\in\sub(\phi)$, for all $n\in\set{0,\dots,m}$, for all $Y_n\overset{\varepsilon : \sigma}{\rightarrow}_*Y
\overset{\varepsilon : \sigma'}{\rightarrow}_*Z$, verifying (if $n<m$) $Z\cup\set{a_{n+1}}$ consistent and reduced, and $Y_{n+1}=\nex(Z)$
$$\psi\in Y\implies a_{n+1}a_{n+2}\dots a_m,N\models\psi$$
where $N=\val(\sigma'\sigma_{n+1}\dots\sigma_m)$.
\end{lem}

\begin{proof}
We do a reverse induction on $n$.

If $n=m$, $Y_n$ is a final state so $Y_n=\emptyset$ or $Y_n=\set{\Omega}$. If $Y_n\overset{\varepsilon : \sigma}{\rightarrow}_*Y$, then $Y=Y_n$ (no outgoing $\varepsilon$-transitions defined from $\emptyset$ or $\set{\Omega}$).
Then if $\psi\in Y$, the only possibility is $\psi=\Omega$, but $a_{n+1}\dots a_m=\varepsilon$, and $\varepsilon,0\models\Omega$, hence the result is true for $n=m$.

Let $n<m$, we assume the result is true for $n+1$, and we take same notations as in the lemma, with $\psi\in Y$.
By definition of $\Delta$, there exists a transition $Y_n\overset{a_{n+1} : \sigma\sigma'}{\longrightarrow}_*\nex(Z)=Y_{n+1}$ in $\aut_\phi$.

We do an induction on the length $k$ of the path $Y\overset{\varepsilon : \sigma'}{\rightarrow}_*Z$.

If $k=0$, then $Y=Z$ is consistent and reduced, so $\psi$ is either atomic or a Next formula.

If $\psi$ is atomic, the only way $Z\cup\set{a_{n+1}}$ can be consistent is if $\psi=a_{n+1}$. In which case we obtain $a_{n+1}\dots a_m,N\models\psi$ without difficulty.

If $\psi=X \varphi$ with $\varphi\in\nex(Z)=Y_{n+1}$, then it corresponds to the case $k=0$. By induction hypothesis (on $n$), $a_{n+2}\dots a_m,N\models\varphi$ ($N$ does not change because $\sigma'$ is empty). Hence $a_{n+1}a_{n+2}\dots a_m,N\models X\varphi$ which shows the result.

If $k>0$, we assume the result is true for $k-1$, and we show it for $k$.
We have $Y\overset{\varepsilon : \sigma'_1}{\rightarrow}Y'\overset{\varepsilon : \sigma'_2}{\rightarrow}_*Z$ with $\sigma'_1\sigma'_2=\sigma'$, and for all $\psi'\in Y', a_{n+1}a_{n+2}\dots a_m,N'\models\psi'$ with $N'=\val(\sigma'_2\sigma_{n+1}\dots\sigma_m)$.

We now look at the different possibilities for the $\varepsilon$-transition $Y\overset{\varepsilon : \sigma'_1}{\rightarrow}Y'$.
Let us first notice that either $N=N'$ or $N=N'+1$: since $\sigma'_1\in\set{\varepsilon,ic,r}^\Gamma$, adding it at the beginning of a sequence can only increment its value by one, or leave it unchanged.

Let $u_{n+1}=a_{n+1}a_{n+2}\dots a_m$.
If $\psi\in Y'$, then $u_{n+1},N'\models\psi$, but $N\geq N'$ so $u_{n+1},N'\models\psi$.

We just need to examine the cases where $\psi\notin Y'$ :
\begin{itemize}
\item If $\psi=\varphi_1\wedge\varphi_2$, $\sigma'_1=\varepsilon$, and $Y'=Y\setminus\set{\psi}\cup\set{\varphi_1,\varphi_2}$,

then $u_{n+1},N\models\varphi_1$ and $u_{n+1}\dots a_m,N\models\varphi_2$, hence $u_{n+1},N\models\psi$.

\item  Other classical cases where $\sigma'_1=\varepsilon$ are similar and come directly from the definition of LTL operators.

\item If $\psi=\varphi_1\UN_j\varphi_2$, $\sigma'_1=\varepsilon$ and $Y'=Y\setminus\set{\psi}\cup\set{\varphi_1,X\psi}$,

then $u_{n+1},N\models\varphi_1$ and $u_{n+1},N\models X\psi$, hence $u_{n+1},N\models\psi$

\item If $\psi=\varphi_1\UN_j\varphi_2$, $\sigma'_1=ic_j$ and $Y'=Y\setminus\set{\psi}\cup\set{X\psi}$,

then $u_{n+1},N'\models X\psi$.

If $\gamma_j$ reaches $N'$ before its first reset in $\sigma'_2\sigma_{n+1}\dots\sigma_m$, then $N=N'+1$, and we can conclude  $u_{n+1},N\models \psi$.

On the contrary, if $N=N'$ and there are strictly less than $N'$ mistakes on $\varphi_1$ before the next occurence of $\varphi_2$, we can allow one more while still respecting the constraint with respect to $N=N'+1$, so $u_{n+1},N\models\psi$.

\item If $\psi=\varphi_1\UN_j\varphi_2$, $\sigma'_1=r_j$ and $Y'=Y\setminus\set{\psi}\cup\set{\varphi_2}$
then $N=N'$, and $u_{n+1},N'\models X\varphi_2$, hence  $u_{n+1},N\models \psi$.
\end{itemize}

Hence we can conclude that for all $k$, $a_{n+1}a_{n+2}\dots a_m,N\models\psi$, which concludes the proof of the lemma.

\end{proof}

Lemma \ref{correct} implies the correctness of the automaton $\aut_\phi$ :\\
Let $Y_0\overset{a_1 : \sigma_1}{\rightarrow}Y_1\overset{a_2 : \sigma_2}{\rightarrow}\dots\overset{a_m : \sigma_m}{\rightarrow}Y_m$ be a valid run of $\aut_\phi$ on $u$ of value $N=\semB{\aut_\phi}$, applying Lemma \ref{correct} with $n=0$ and $Y=Y_0=\set{\phi}$ gives us $(u,N)\models\phi$. Hence $\sph\leq\semB{\aut_\phi}$.

Conversely, let $N=\sph(u)$, then $(u,N)\models\phi$ so by definition of $\aut_\phi$, it is straightforward to verify that there exists an accepting run of $\aut_\phi$ over $u$ of value $\leq N$ (each counter $\gamma_i$ doing at most $N$ mistakes relative to operator $\UN_i$). Hence $\semB{\aut_\phi}\leq\sph$.

We finally get $\semB{\aut_\phi}=\sph$, the automaton $\aut_\phi$ computes indeed the exact value of function $\sph$ (and so we have obviously $\semB{\aut_\phi}\approx\sph$).
%
%\begin{cor}
%Boundedness is decidable for $\ltlq$-formulae.
%\end{cor}
%
%\begin{proof}
%Since boundedness of $B$-automata is decidable by \cite{Col09}, this translation from $\ltlq$ to $B$-automata allows us to decide boundedness for $\ltlq$-formulae. That is to say, there is an algorithm taking as input an $\ltlq$-formula $\varphi$, and deciding whether $\sem{\varphi}$ is bounded, i.e. $\sem{\varphi}\approx 0$. In fact we can even do better, and decide $\sem{\varphi}\preccurlyeq\sem{\psi}$ for any two $\ltlq$-formulae $\varphi$ and $\psi$.
%\end{proof}

\subsection*{Contraction of actions}
If we want to obtain a $B$-automaton as defined in Section \ref{sec:aut}, with atomic actions on transitions, we can proceed as follow.

We replace every action $\sigma\in\set{\ic,\e,\r}^*$ by the maximal letter atomic action $\max(\sigma)$ occuring in it, with respect to the order $\e<\ic<\r$. For instance $\ic\r\ic\ic$ will be changed in $\r$. Let $K$ be the maximal number of consecutive increments in such an action $\sigma$, that is to say $$K=\max\set{\val_B(\sigma) : (p,a,\sigma,q)\in\Delta}.$$

Let $\aut'$ be the automaton obtained from $\aut_\phi$ by replacing each action $\sigma$ by $\sigma'=\max(\sigma)$.
Runs of $\aut_\phi$ and $\aut'$ are in a one-to-one correspondance in a canonic way: only counter actions have changed.
Let $\rho$ be a run of $\aut_\phi$ and $\rho'$ be the corresponding run of $\aut'$.

First, remark that $\rho'$ always perform less increments than $\rho$ (going from $\sigma$ to $\sigma'$ only remove increments), so $\val_B(\rho')\leq\val_B(\rho)$.

Moreover, when we go from $\sigma'$ to $\sigma$, we can add at most $K$ increments before or after each reset, so between two resets of $\rho$ (or edge of the word), we have at most $2K$ increments for each increment in $\rho'$. Let $\alpha(n)=2Kn+2K$, we obtain $\val_B(\rho)\leq \alpha(\val_B(\rho'))$.

Let $u\in\A^*$ and $\rho$ a run of $\aut_\phi$ such that $\val_B(\rho)=\semB{\aut_\phi}(u)$. Then we saw that there a run $\rho'$ of $\aut'$ with $\val_B(\rho')\leq\val_B(\rho)$. Thus we obtain $\semB{\aut'}\leq \semB{\aut}$.

Conversely, let $\rho'$ be a run of $\aut'$ such that $\val_B(\rho')=\semB{\aut'}(u)$. Then we saw that there is a run $\rho$ of $\aut$ with $\val_B(\rho)\leq_\alpha\val_B(\rho)$. Thus we obtain $\semB{\aut'}\preccurlyeq_\alpha \semB{\aut_\phi}$.

In the end, we get $\semB{\aut'}\approx\semB{\aut_\phi}=\sem{\phi}$, and $\aut'$ is a $B$-automaton with atomic actions.

The results are summed up in the following theorem:

\begin{thm}\label{thm:ltlreg}
Let $\varphi$ be an $\ltlq$-formula, we showed that $\sem{\varphi}^\approx$ is recognized by a $B$-automaton, and so $\sem{\varphi}^\approx$ is a regular cost function.
If we authorize non-atomic actions on transitions, we can build a $B$-automaton that preserves the exact semantic of $\varphi$, not using approximation $\approx$.
\end{thm}

Since by \cite{Col09}, we can decide whether a function recognized by $B$-automaton is bounded, or even compare such functions with respect to $\preccurlyeq$, we get the following corollary:

\begin{cor}\label{cor:ltldec}
Let $\varphi$ and $\psi$ be two $\ltlq$-formulae, we can decide whether $\sem{\varphi}$ is bounded, and more generally whether $\sem{\varphi}\preccurlyeq\sem{\psi}$ holds.
\end{cor}

Notice that deciding whether $\sem{\varphi}$ is bounded amounts to decide whether $\sem{\varphi}\preccurlyeq 0$ (where $0$ is the function mapping every word to $0$). Notice that boundedness of a formula corresponds to ``uniform validity'' of the formula: a formula is bounded if it can accepts every input, within a uniform bound $N$. In particular, a classical LTL formula is bounded if and only if it is true on all words.
This is illustrated in Example \ref{exbounded}. 

\begin{exa}\label{exbounded}
We give two examples on alphabet $\{a,b\}$: let $\varphi=(b\vee X a\vee XF a)\UN\Omega$, and $\psi=(a\vee X a\vee XF a)\UN\Omega$.
Then $\sem{\varphi}$ is bounded by $2$: the subformula $(b\vee X a\vee XF a)$ fails if the remaining suffix is in $ab^+$ (which happens at most once), or if we are on the last letter and it is $a$. On the other hand, $\sem{\psi}$ is unbounded, because $\sem{\psi}(b^n)=n$ for all $n$.
\end{exa}

From \cite{Col09}, deciding whether $\sem{\varphi}\preccurlyeq\sem{\psi}$ requires to build an $S$-automaton recognizing $\sem{\varphi}^\approx$, and a $B$-automaton recognizing $\sem{\psi}^\approx$. This means that to test boundedness of a formula $\varphi$, we want to obtain an $S$-automaton recognizing $\sem{\varphi}^\approx$.
Moreover, the standard algorithm translating between $B$- and $S$-automata is in EXPSPACE, because it uses the underlying stabilization semigroup, possibly containing exponentially many elements, compared to the number of states of automata.

To reduce the complexity of these two decision problems (boundedness and comparison of $\ltlq$-formulae), it is therefore useful to transform $\ltlq$-formulae directly into $S$-automata.

\section{From $\ltlq$ to $S$-automata}\label{sec:ltlSaut}

In this section, we give a translation from $\ltlq$-formulae to the model of the $S$-automata.
This will allow us to show that the boundedness problem for $\ltlq$-formulae is PSPACE-complete.

\subsection{The logic $\nltlq$}
%
%Afin de définir de manière naturelle un $S$-automate à partir d'une formule de \ltlq, on veut commencer par renverser la sémantique de cette formule.
In order to naturally define a $S$-automaton from a $\ltlq$-formula, we will start be reversing the semantic of this formula.

Let $\nltlq$ be the logic defined by the following grammar:
$$\varphi:=a~|~\varphi\wedge\varphi~|~\varphi\vee\varphi~|~X\varphi~|~\varphi U\varphi|~~\varphi\RgN\varphi~|~\Omega$$ 

We want such a formula to be obtained by negating a $\ltlq$-formula, and then pushing negations to the leaves. This is why we need a dual operator to $\UN$, which is $\RgN$.
We want its semantic to be such that $(\neg\varphi)\RgN(\neg\psi)$ is equivalent to $\neg(\varphi\UN\psi)$. 
That is to say, the semantic of $\RgN$ is defined by: $(u,n,i)\models \varphi\RgN\psi $ if for all $j>i$, either $(u,n,j)\models\psi$, or there are at least $n$ positions $i\leq j'<j$ such that $(u,n,j')\models\varphi$.
Other operators have same semantics as in $\ltlq$.

We can notice that if $\varphi$ is a $\ltlq$-formula, then $\neg\varphi$ is equivalent to a $\nltlq$-formula, by pushing negations to the leaves.

If $\varphi$ is a $\nltlq$-formula, we define the cost function $\nsem{\varphi}$ recognized by $\varphi$ by
$$\nsem{\varphi}(u)=\sup\set{n\in\N : (u,n)\models\varphi}.$$

\begin{lem}
Let $\varphi$ be a $\ltlq$-formula, then $\nsem{\neg\varphi}\approx\sem{\varphi}$.
\end{lem}

\begin{proof}
Let $\varphi$ be a $\ltlq$-formula, and $u\in\A^*$.
If $\sem{\varphi}(u)=\infty$, then for all $n\in\N$, $(u,n)\models\neg\varphi$, hence $\nsem{\neg\varphi}=\infty$.
Otherwise, let $n=\sem{\varphi}(u)\in\N$, then $(u,n)\models\varphi$ and $(u,n+1)\not\models\varphi$.
Thus we have $(u,n)\not\models\neg\varphi$ and $(u,n+1)\models\neg\varphi$, this implies $\nsem{\neg\varphi}=n+1$.

This is enough to conclude $\nsem{\neg\varphi}\approx\sem{\varphi}$.
\end{proof}
Going from an $\ltlq$-formula to its negation in $\ltlq$ can be done by a linear time algorithm: it suffices to push negations to the leaves, replacing each operator by the dual one (with possible addition of $\Omega$). Thus it suffices to build the wanted $S$-automaton from an $\nltlq$-formula.

\subsection{From $\nltlq$ to $S$-automata}% hiérarchiques}
Let $\phi$ be an $\nltlq$-formula, with $k$ $\RgN$-operators, labelled $\RgN_1,\RgN_2,\dots,\RgN_k$.

We can build a $S$-automaton $\aut_\phi$ as before, with counters $\set{\gamma_1,\dots,\gamma_k}$, by remembering subformulae of $\phi$ as constraints in states.
The states of $\aut_\phi$ are again $Q=2^{\sub(\phi)}$, with $\set{\phi}$ as initial state. However, this time, the final states are every state $Y$such that for all $\varphi\in Y$, we have $\varphi=\Omega$ or $\varphi$ is of the form $\varphi_1\RgN\varphi_2$. Indeed, we have $(\epsilon,0)\models\varphi_1\RgN\varphi_2$, for any formulae $\varphi_1$ and $\varphi_2$.
Then, the main new feature is how we deal with operators $\RgN$ in the table of $\e$-transitions between pseudo-states.

Let $Y$ be a pseudo-state (or a real state), and $\psi$ be a non-reduced formula of maximal size in $Y$. If $\psi$ is not of the form $\varphi_1\RgN_j\varphi_2$, then we add the same transitions as in Section \ref{descB}.

Otherwise, if $\psi=\varphi_1\RgN_j\varphi_2$, we add the following transitions:
$$\left\{\begin{array}{l}
Y\overset{\e : \i_j}{\longrightarrow}Y\setminus\set{\psi}\cup\set{\varphi_1,\varphi_2,X\psi}\text{ (we count one occurence of $\varphi_1$, and $\varphi_2$ has to be seen)}\\
Y\overset{\e : \e}{\longrightarrow}Y\setminus\set{\psi}\cup\set{\varphi_2,X\psi}\text{ (we see $\varphi_2$ without $\varphi_1$)}\\
Y\overset{\e : \ccr_j}{\longrightarrow}Y\setminus\set{\psi}\text{ (if $\varphi_2$ cannot be proved, we perform $\ccr$ to guarantee a lot of $\varphi_1$ before)}\\
\end{array}\right.$$

The proof of correctness is very similar to the one for $B$-automata in Section \ref{corrB}, so we omit it here. The main intuition is to that all transitions keep track of constraints in a sound way.

As before, we could build the transition table of $\aut_\phi$ by contracting all $\e$-transitions, and verify that the resulting $S$-automaton recognizes $\nsem{\phi}$. 
For our current purpose, it is not necessary to perform this contraction, we can think of $\aut_\phi$ as having all the $\e$-transitions and pseudo-states described in its construction.

However, contracting sequences of $S$-actions will be useful in another context, to describe the PSPACE algorithm. So the next section describes how such sequences can be contracted.

\subsection{Semigroup of $S$-actions}\label{semSact}
We will explicit how to contract $S$-actions, by using a stabilization semigroup $\semi$ which contains all the necessary information about how to compose these actions.
The product operation in $\semi$ reflects the concatenation of $S$-actions.
The stabilization operation $\sharp$ corresponds to repeating the same action a lot of times, for instance as it can be done in a cycle of the automaton.
The element $\omega$ stands for a ``big value'', which can be made arbitrarily large, by repeating element $\i$ a lot of times.
The element $\bot$ represents a fail of the run, when the automaton tries to perform action $\ccr$ on a small counter value.
The elements of $\semi$ are gathered in the set $S=\set{\omega, \i, \e, \r, \ccr\omega, \ccr, \bot}$, and ordered by $\omega\leq \i\leq \e\leq (\r/\ccr\omega)\leq\ccr\leq\bot$. 
This order reflects a preference for the $S$-automaton: between two actions $\sigma\leq\sigma'$, it is always better to choose $\sigma$ in any context, when aiming for a big $S$-value. 
This explains why actions $\r$ and $\ccr\omega$ are not comparable: the best choice can depend on the context.
Indeed, in an empty context, action $\r$ is better (it yields value $\infty$ while $\ccr\omega$ yields value $0$), but in a context $C[x]=\omega\cdot x\cdot\ccr$, it is better to choose $\ccr\omega$ (yielding value $\infty$) than $\r$ (yielding value $0$). Product and stabilization operations in $\semi$ are explicited in the following array: 
$$
\begin{array}{|c||c|c|c|c|c|c|c||c|}
\hhline{-||-------||-}
 {\cdot}& \omega & \i & \e & \r & \ccr\omega & \ccr & \bot & \cdot^\sharp\\
\hhline{=::=======::=}
 \omega & \omega & \omega & \omega & \r & \omega & \r & \bot & \omega\\
\hhline{-||-------||-}
\i & \omega & \i & \i  & \r & \ccr\omega & \ccr  & \bot & \omega \\
\hhline{-||-------||-}
\e & \omega & \i & \e  & \r & \ccr\omega & \ccr & \bot & \e\\
\hhline{-||-------||-}
 \r & \omega & \r & \r & \r & \bot & \bot & \bot & \r \\
\hhline{-||-------||-}
 \ccr\omega & \ccr\omega & \ccr\omega & \ccr\omega & \ccr & \ccr\omega & \ccr & \bot &\ccr\omega \\
\hhline{-||-------||-}
 \ccr & \ccr\omega & \ccr & \ccr & \ccr & \bot & \bot & \bot &\\
\hhline{-||-------||-}
 \bot & \bot & \bot & \bot & \bot & \bot & \bot & \bot & \bot\\
\hhline{-||-------||-}

\end{array}
$$
Notice that $\ccr^\sharp$ is undefined because $\ccr$ is not idempotent.

If $\Gamma$ is a set of counters, then we denote by $\semi^\Gamma$ the product stabilization semigroup with underlying set $S^\Gamma$, where all operations are performed component-wise.
If $\sigma\in\semi^\Gamma$ and $\gamma\in\Gamma$, we will note $\sigma_\gamma$ the projection of $\sigma$ on counter $\gamma$. When some components are not specified, the default value is $\e$. For instance if $\Gamma=\set{1,2}$, we can write $\i_1$ for $(\i,\e)$ and $\ccr\omega_2$ for $(\e,\ccr\omega)$.
\subsection{Decision algorithm in polynomial space}\label{sec:algo}
It was shown in \cite{SC85} that satisfiability of classical LTL-formula is a PSPACE-complete problem. To obtain a PSPACE algorithm, an equivalent automaton is generated on-the-fly, and an accepting run of this automaton is guessed, while only information about the current state is remembered.
Transition labels can be ignored, as only the existence of an accepting run is of interest, we do not care about which word can be accepted.

We want here to generalize this approach to cost functions: the problem is to explore the automaton $\aut_\phi$ described earlier, but without keeping the whole automaton in the memory of the algorithm, in order to use only polynomial space with respect to the size of $\phi$.
We now want to decide whether the function $\nsem{\phi}$ described by $\phi$ is bounded. This generalizes the satisfiability problem for LTL, since an LTL formula $\varphi$ is satisfiable if and only if $\sem{\neg\varphi}$ is unbounded, that is to say $\nsem{\varphi}$ is unbounded.

To do so, we will look for a witness of the fact that $\semS{\aut_\phi}$ is unbounded.
We have to face here an additional challenge compared to the classical case: it is not enough to find an accepting path, we have to find a family of accepting paths with arbitrary high values. This means that while we can forget the letters labelling transitions, we have to pay attention to counter actions.

We will need to keep information about counter values along the way. The principles that the algorithm has to respect for each counter $\gamma\in\Gamma$ are the following:
\begin{itemize}
\item Every action $\ccr_\gamma$ must follow an action $\omega_\gamma$, which represents a big number of increments $\i_\gamma$.
\item The only way to obtain $\omega_\gamma$ is to go through a cycle containg at least one $\i_\gamma$, and only actions $\i$ and $\e$ for $\gamma$.
\end{itemize}
Thus the aim is to describe a non-determinist algorithm that guesses a path in the automaton, as well as states that will be visited twice (in order to create cycles). These states will be called ``control points''.
We explain how the algorithm works via the following example:

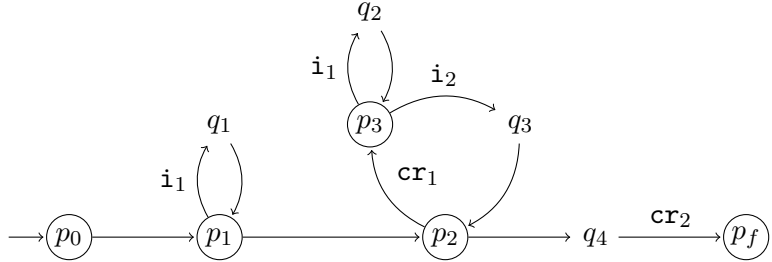
\begin{center}
\begin{tikzpicture}[shorten >=1pt,node distance=2cm,on grid,auto,initial text=,
every state/.style={inner sep=0pt,minimum size=6mm},accepting/.style=accepting by arrow]

    \node[state,initial]	(p_0) {$p_0$};;
    \node[state] 	(p_1) [right=of p_0] {$p_1$};
    \node 	(q_1) [above=1.5cm of p_1] {$q_1$};
    \node[state] 	(p_2) [right=3cm of p_1] {$p_2$};
    \node[state] 	(p_3) [above left=1.5cm and 1cm of p_2] {$p_3$};
    \node 	(q_2) [above=1.5cm of p_3] {$q_2$};
    \node 	(q_3) [right=of p_3] {$q_3$};
    \node 	(q_4) [right=of p_2] {$q_4$};
    \node[state] 	(p_f) [right=of q_4] {$p_f$};
    \path[->]
    (p_0) edge (p_1)
    (p_1) edge[bend left] node {$\i_1$} (q_1)
    		edge (p_2)
    (q_1) edge[bend left]  (p_1)
    (p_2) edge[bend left] node[anchor=west,above right]{$\ccr_1$} (p_3)
    		edge (q_4)
    (p_3) edge[bend left] node {$\i_1$} (q_2)
    		edge[bend left] node {$\i_2$} (q_3)
    (q_2) edge[bend left]  (p_3)
    		
    (q_3)	edge[bend left]  (p_2)
    (q_4) edge node{$\ccr_2$} (p_f);

\end{tikzpicture}
\end{center}

If $p_0\in\mathit{In}$ and $p_f\in\mathit{Fin}$, the existence of this path in the automaton $\aut_\phi$ is a witness that $\semS{\aut_\phi}$ is unbounded.
The aim of the algorithm is to find such a path, by guessing the beginning of each cycle, and contracting actions between two control points. In the example, the control points are $p_1, p_2$ and $p_3$.
At any time, the memory of the algorithm contains a sequence $m,\sigma_1,p_1,\sigma_2,p_2,\dots, \sigma_m,p_m$, where for all $i\in[1,m]$, action $\sigma_i$ is in $\semi^\Gamma$, and $p_i\in Q$. Moreover, $p_m$ is always the current state of the run of $\aut_\phi$, and $m\leq|\Gamma|+1$.
States $(p_i)_{i<m}$ are the current control points, i.e. starts of cycles that the run is currently using.  
They have to be closed in the future for the algorithm to end, and the last to be opened has to be the first closed, so that we get properly nested cycles. 
When a cycle is closed, operator $\sharp$ is applied to the action performed in the cycle, and then the product operation is used to concatenate this action with the one of the new innermost current cycle.

The aglorithm starts with memory $0,\e,p_0$ with $p_0\in\mathit{In}$. In general $p_0$ can be chosen in a nondeterministic way, but here the unique initial state is $\set{\phi}$.
Transitions are used on-the-fly: at any position, we guess a transition (available transitions depend only on formulae $\varphi$ appearing in the current state), and we update the memory accordingly.
The algorithm ends and outputs ``unbounded'' if the memory only contains $0,\sigma,p_f$, with $p_f\in\mathit{Fin}$, and for all $\gamma\in\Gamma$, $\sigma_\gamma\notin\set{\ccr,\ccr\omega,\bot}$. 

We can remark that the condition $m\leq|\Gamma|+1$ limits the number of nested cycles to $|\Gamma|$.
This guarantees a memory space polynomial in $|\phi|$, since $|\Gamma|\leq|\phi|$ (every counter comes from an operator $\RgN$ of $\phi$).

We come back to the above example, and we describe in the following table the successive statuses of the memory, while the algorithm guesses the wanted witness. For convenience, we focus here on the status of the memory when the algorithm passes through states $q_1,q_2,q_3,q_4,p_f$:

$$\begin{array}{c|c|cccccc}
 & m &\\
 \hline
 q_1 & 1 & \e & p_1 & \i_1 & q_1 \\
 \hline
 q_2 & 2 & \omega_1 & p_2 & \ccr_1 & p_3 & \i_1 & q_2 \\
  \hline
 q_3 & 1 & \omega_1 & p_2 & (\ccr\omega,\i) & q_3\\
  \hline
 q_4 & 0 & (\r,\omega) & q_4\\
  \hline
 p_f & 0 & (\r,\r) & p_f\\
  \hline
\end{array}$$

After passing through $q_1$, the algorithm goes back to $p_1$, and closes a cycle with only action $\i_1$. Thus it gets action $\i_1^\sharp=\omega_1$.
Then the control points $p_2$ and $p_3$ are opened, with an action $\ccr_1$ in-between. They are followed by an action $\i_1$, as we can see in state $q_2$.
After state $q_2$, when the algorithm goes back to $p_3$, action $\i_2$ is stabilized, yielding $\omega_2$. 
It is then concatenated with the previous $\ccr_1$ yielding action $(\ccr\omega,\i)$ that we can see in $q_3$.
When the run goes back to $p_2$ and closes the external cycle, this $(\ccr\omega,\i)$ is stabilized into $(\ccr\omega,\i)^\sharp=(\ccr\omega,\omega)$, and concatened with $\omega_1$, yielding the $(\r,\omega)$ action that we can see in $q_4$. 
In the end, concatenation with the final $\ccr_2$ yields global action $(\r,\r)$ in $p_f$. By the acceptance condition of the algorithm, it can stop there and output ``unbounded'', since $\r\notin\set{\ccr,\ccr\omega,\bot}$ and $p_f\in\mathit{Fin}$.

\subsection{Complexity and correctness of the algorithm}

\begin{lem}
The algorithm described in Section \ref{sec:algo} has a space complexity polynomial in $|\phi|$.
\end{lem}

\begin{proof}
We start by precising how the formula $\phi$ is given as input to the algorithm.
Such a formula can be represented by a tree, whose nodes are operators, and whose leaves are atoms. For instance the formula $(a\RgN b)U ((X X a) \vee (b\RgN\Omega))$ will be coded by the following tree:

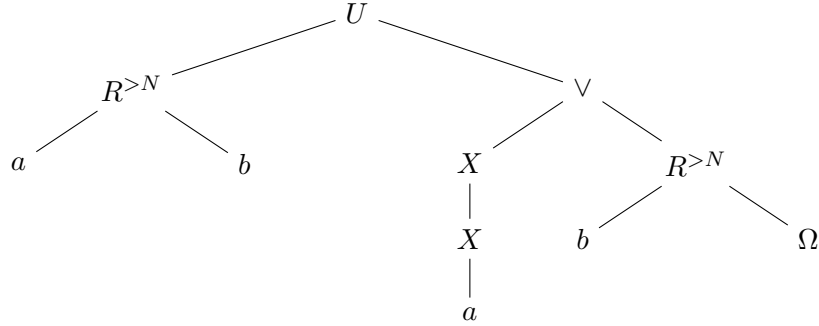
\begin{center}
\begin{tikzpicture}[level distance= 1cm,level 1/.style={sibling distance=6cm},level 2/.style={sibling distance=3cm}]
  \node {$U$}
                child {node {$\RgN$}
                	child {node {$a$}}
                	child {node {$b$}}}
                child {node {$\vee$}
                    child {node {$X$} child {node {$X$} child {node {$a$}}}}
             	child {node {$\RgN$}
                	child {node {$b$}}
                	child {node {$\Omega$}}
             			}}  ;
\end{tikzpicture}
\end{center}

This way, each subformula of $\phi$ corresponds to a node in this tree. A set of subformula is therefore just a set of nodes, and every state of the automaton can be encoded by a tuple $(n_1,n_2,\dots,n_t)$, where every $n_i$ encodes the position of a node of $\phi$ (it is easy to see that such an encoding is polynomial in the size of $\phi$). 
Thus, the encoding of a state takes a polynomial space with respect to the size of the input tree, which is $|\phi|$.

We can remark that this is still true when adding pseudo-states, because those are subsets of $\sub(\phi)\cup\set{X\varphi: \varphi\in\sub(\phi)}$. Therefore, the encoding of a pseudo-state is at most twice as long as the encoding of a real state, so it still takes only polynomial space.

The encoding of an element in $|\semi|$ takes constant space, so it takes a space linear in $|\Gamma|$ to encode an element of $\semi^\Gamma$. Since every counter in $\aut_\phi$ comes from an operator $\RgN$ of $\phi$, we have $|\Gamma|\leq |\phi|$.
Therefore, each element of $\semi^\Gamma$ takes a space linear in $|\phi|$ in the memory.

Finally, we have $m\leq|\Gamma|$, so a space logarithmic in $|\phi|$ is enough to store $m$. At any time, we will have at most $m$ pseudo-states and $m$ elements of $\semi^\Gamma$ in the memory, each one taking polynomial space in $|\phi|$. We can conclude that the whole sequence occupies a space which is only polynomial in $|\phi|$.
\end{proof}
\begin{lem}
The algorithm is correct, that is to say it outputs ``unbounded'' if and only if $\nsem{\phi}$ is unbounded.
\end{lem}
\begin{proof}

It is easy to show that if the algorithm outputs ``unbounded'', then $\nsem{\phi}$ is unbounded. Indeed, the algorithm describes a path (with cycles) in $\aut_\phi$. It is straightforward to show that if every cycle is taken $n$ times, the value of the resulting run is at least $n$. Therefore, the path found by the algorithm describes a family of runs of arbitrarily high value, so we can conclude $\nsem{\phi}$ is unbounded.
\smallskip

We now show the converse: we assume $\semS{\aut_\phi}$ is unbounded, and we want to show that there exists a witness path that can be found by the algorithm.
To do this, we define for all $S$-automaton $\aut$ the stabilization semigroup $\semi_\aut=\perm{S_\aut,\cdot,\sharp,\leq}$, whose elements represent sets of partial runs of $\aut$. This construction parallels the one in \cite{Col09}.

A partial run from state $p$ to state $q$ performing global action $\sigma$ will be represented by the element $(p,\sigma,q)\in Q\times S^\Gamma\times Q$.

If $(p,\sigma,q)$ and $(p',\sigma',q')$ are two elements of $Q\times S^\Gamma\times Q$, we will say that $(p,\sigma,q)\leq(p',\sigma',q')$ if $(p,q)=(p,q')$ and $\sigma\leq\sigma'$, for the ordered on $S^\Gamma$ defined in Section \ref{semSact}.

Let $E\subseteq {Q\times S^\Gamma\times Q}$, we will denote by $E\downarrow=\set{e\leq e' : e'\in E}$ the downards-closure of $E$.
Let $S_\aut=2^{Q\times S^\Gamma\times Q}\downarrow$ be the set of downwards-closed elements of $\semi_\aut$. Each element $E$ of $S_\aut$ represents a set of runs. The downwards-closure operation reflects the fact that we consider that the automaton is allowed to perform actions that are less efficient than the real ones: it does not change its global semantic.

Product and stabilization operation in $\semi_\aut$ are defined by:
$$\begin{array}{c}
E\cdot F=\set{(p,act_1\cdot \sigma_2,r) : (p,\sigma_1,q)\in E,(q,\sigma_2,r)\in F}\downarrow\\
E^\sharp=\set{(p,\sigma_1\cdot\sigma_e^\sharp\cdot\sigma_2,r) : (p,\sigma_1,q),(q,\sigma_e,q),(q,\sigma_2,r)\in E}\downarrow.
\end{array}$$
%
%Consequently, only elements of the form $(p,\sigma,p)$ are idempotents (with $\sigma_\gamma\neq\ccr$ for all $\gamma$). On définit la stabilisation par $(p,\sigma,p)^\sharp=(p,\sigma^\sharp,p)$, qui correspond à répéter un cycle de $p$ vers $p$ un grand nombre de fois.
Notice that each element $E$ describes a set of partial runs, and therefore, witnesses of accepting runs are described by the following subset of $S_\aut$:
%We also define the accepting ideal of $\semi_\aut$ by:
$$I'=\set{E\in S_\aut : \exists (p,\sigma,q)\in E, p\in\mathit{In}, q\in\mathit{Fin}, \forall \gamma\in\Gamma,\sigma_\gamma\notin\set{\ccr,\ccr\omega,\bot}},$$ together with the morphism $h:\A\to \semi_\aut$ by $h(a)=\set{(p,\sigma,q) : (p,a,\sigma,q)\in\Delta_\aut}\downarrow$.

Since accepting ideals are defined as elements of big value, we take the accepting ideal to be $I=S_\aut\setminus I'$.

It is not hard to verify that $S_\aut, h,I$ recognizes the cost function $\semS{\aut}$ (see \cite{Col09} for more details).
Consequently, $\semS{\aut}$ is unbounded if and only if $\perm{h(\A)}^\sharp\cap I\neq\emptyset$, i.e. there is an element of $I$ that can be obtained from $h(\A)$ via product and stabilization operations. Indeed, such an element can be described by a \se\ $e$ well-formed for $\semi_\aut$, with $\eval(e)\in I$, witnessing a sequence of words $(e(n))_{n\in\N}$ of unbounded value.
We now apply this construction to the automaton $\aut_\phi$ obtained from $\phi$.

Since we assumed $\semS{\aut_\phi}$ is unbounded, there exists a \se\ $e$, well-formed for $\semi_{\aut_\phi}$, such that $\eval(e)\in I$. It remains to show that $e$ does not need more than $|\Gamma|$ nested stabilization operators.

Let us assume that $e$ contains at least $k=|\Gamma|+1$ nested stabilization operators $\sharp_1,\dots,\sharp_k$, applied to \se s $e_1,\dots, e_k$ respectively.
Let $(p_0,\sigma_f,p_f)\in \eval(e)$, witnessing the fact that $\eval(e)\in I$, that is to say $p_0\in\mathit{In}, p_f\in\Fin$, and for all $\gamma\in\Gamma, \sigma_\gamma\notin\set{\ccr,\ccr\omega,\bot}$.
We will say that a stabilization operator $\sharp_i$ is \emph{useful} if the element described by the \se\ obtained from $e$ by removing $\sharp_i$ does not contain $(p_0,\sigma_f,p_f)$.
We show by induction on $|\Gamma|$ that at least one of the operators $\sharp_1,\dots,\sharp_k$ is not useful.
If $|\Gamma|=0$, then $\phi$ is a classic LTL formula, the automaton computes the characteristic function of a regular language, and stabilization is just the identity on $\semi_{\aut_\phi}$, so no stabilization operator can be useful. We now assume $|\Gamma|\geq 1$.
Let $\sharp_k$ be the outmost stabilization operator in $e$. Therefore, we can write $e=x\cdot e_k^{\sharp_k}\cdot y$, where $x$ and $y$ are \se s. let $E=\eval(e)=\eval(x)\eval(e_k^\sharp)\eval(y)=XE_k^\sharp Y$.
We assume that $\sharp_k$ is useful (otherwise we get the wanted result).
By definition of the product of $\semi_{\aut_\phi}$, only one of the elements of $E_k^\sharp$ is used to obtain $(p_0,\sigma_f,p_f)\in E$. By definition of $\sharp$, this element is of the form $(p,\sigma,r)$, with $(p,\sigma_1,q),(q,\sigma_e,q),(q,\sigma_2,r)\in E$ and $\sigma\leq \sigma_1\cdot\sigma_e^\sharp\cdot\sigma_2$. Since $\sharp_k$ is useful, we must have $\sigma_e^\sharp\neq\sigma_e$, so there exists $\gamma\in\Gamma$ such that $(\sigma_e)_\gamma=\i$.
Moreover, by definitions of the operations of $\semi_{\aut_\phi}$, the \se\ $e_k$ can not contain any useful stabilisation on counter $\gamma$.
Therefore, we are left with the $k-1$ stabilisations in $e_k$, and $|\Gamma|-1$ available counters, since $\gamma$ is no longer influenced by stabilizations. This concludes the proof by induction.
%
 %Cela n'est possible que si $E$ contient un élément de la forme $(p,\i,p)$, d'après la définition de la stabilisation. Par conséquent, $E^\sharp$ contiendra $(p,\omega,p)$, et par cl\^oture  l'une des composantes $\gamma$ de $\sigma$ est égale à $\i$. La stabilisation transforme ce $\i$ en $\omega$, et il n'existe ensuite plus de moyen de revenir sur $\i$ par produit ou stabilisation, d'après le tableau d'opérations du semigroupe $S$. Ceci signifie que l'une des $|\Gamma|+1$ stabilisations imbriquées de $e$ est inutile, et peut \^etre supprimée de $e$. 

We can conclude that $e$ is equivalent (with respect to $\eval$) to an \se\ $e'$ with at most $|\Gamma|$ nested stabilizations.
The fact that $\eval(e')\in I$ guarantees us the existence of a path in $\aut_\phi$ that can be found by our algorithm, since it contains at most $|\Gamma|$ nested cycles. This concludes the proof of the Lemma.
\end{proof}
\begin{thm}\label{pspace-ltlq}
Given an $\ltlq$-formula $\phi$, the problem of deciding whether $\sem{\phi}$ is bounded is PSPACE-complete.
%More generally, given two $\ltlq$-formulae $\phi$ and $\psi$, the problem of deciding whether $\sem{\phi}\preccuryleq\sem{\psi}$ is PSPACE-complete.
\end{thm}
\begin{proof}
We saw that there exists a PSPACE algorithm solving this problem. We start be negating $\phi$ to obtain a formula $\phi'$ of $\nltlq$ (this is done in linear time). We then describe the transition table of the $S$-automaton $\aut_{\phi'}$, and explore this automaton on-the-fly, using only polynomial space. This way we can guess a path witnessing unboundedness of $\nsem{\phi'}$, if such a path exists.

To show that the problem is PSPACE-hard, it suffices to remark that classical LTL satisfiability is a particular case of $\ltlq$ boundedness, and that LTL satisfiability is PSPACE-hard \cite{SC85}. Indeed, if $\phi$ is a classical LTL-formula, we can see $\neg\phi$ as a formula of $\ltlq$, and we get that ``$\sem{\neg\phi}$ bounded'' is equivalent to ``$L(\phi)=\emptyset$''.
\end{proof}

We showed that generalisation of LTL into $\ltlq$ does not increase the computational complexity of the satisfiability/boundedness problem. This result is encouraging, since it allows us to treat a more general problem, without paying anything in terms of computational resources.

\section{Syntactic congruence on \ose s}\label{synccong}
We remind that as in the case of languages, stabilization semigroups recognize exactly regular cost functions, and there exists a quotient-wise minimal stabilization semigroup for each regular cost function \cite{CKL}.

In standard theory, it is equivalent for a regular language to be described by an LTL-formula, or to be recognized by an aperiodic semigroup. Is it still the case in the framework of regular cost functions? To answer this question we first need to develop a little further the algebraic theory of regular cost functions.

\subsection{Syntactic congruence}
In standard theory of languages, we can go from a description of a regular language $L$ to a description of its syntactic monoid via the syntactic congruence. Moreover, when the language is not regular, we get an infinite monoid, so this equivalence can be used to ``test'' regularity of a language.

The main idea behind this equivalence is to identify words $u$ and $v$ if they ``behave the same'' relatively to the language $L$, i.e. $L$ cannot separate $u$ from $v$ in any context : $\forall(x,y), xuy\in L\Leftrightarrow xvy\in L$.
\smallskip

The aim here is to define an analog to the syntactic congruence, but for regular cost functions instead of regular languages. 
Since cost functions look at quantitative aspects of words, the notions of ``element'' and ``context'' have to contain quantitative information : we want to be able to say things like ``words with a lot of $a$'s behave the same as words with a few $a$'s''.

That is why we will not define our equivalence over words, but over \se s, which are a way to describe words with quantitative information.

\subsection{\se s}\label{se}
We first define general \se s as in \cite{Has90} and \cite{CKL} by just adding an operator $\sharp$ to words in order to repeat a subexpression ``a lot of times''. This differs from the stabilization monoid definition, in which the $\sharp$-operator can only be applied to specific elements (idempotents).
\bigskip

The set $\expr$ of \se s on an alphabet $\A$ is defined as follows:
$$e:=a\in\A~|~ee~|~e^\sharp$$

%If $e$ is a \se~and $n\in\N$, then we defined $e(n)$ to be the word on $\A$ obtained by replacing all occurences of $\sharp$ by $n$ in $e$ (exponential for the concatenation).
%Let $\bowtien$ be the equivalence relation over sequences defined by $a_n\bowtien b_n$ iff $\limsup a_n=\infty \Leftrightarrow \limsup b_n=\infty$.
%If $f$ is a regular cost function, let $\sim_F$ be the equivalence relation over $\expr$ defined by
%
%$$e\sim_F e'\text{ iff } \forall C[x]\in\ce,  f(C[e](n))\bowtien f(C[e'](n))$$
%
%\begin{rem}
%This equivalence relation extends the language syntactic congruence : if $f=\chi_L$ for some language $L\subseteq\A^*$, then for any $u,v\in\A^*$, $u\sim_Fv\Leftrightarrow u\sim_Lv$.
%\end{rem}
%

If we choose a stabilization semigroup $\semi=\perm{S,\cdot,\leq,\sharp}$ together with a function $h:\A\rightarrow S$, the evaluation function $\eval$ : $\expr\to\semi$ is defined inductively by $\eval(a)=h(a), \eval(ee')=\eval(e)\cdot\eval(e')$, and $\eval(e^\sharp)=\eval(e)^\sharp$ ($\eval(e)$ has to be idempotent). We say that $e$ is \intro{well-formed for $\semi$} if $\eval(e)$ exists. Intuitively, it means that $\sharp$ was applied to subexpressions that correspond to idempotent elements in $\semi$.

If $f^\approx$ is a regular cost function, $e$ is \intro{well-formed for $f$} iff $e$ is well-formed for the minimal stabilization semigroup of $f^\approx$.

\begin{exa}\label{pair}
Let $f$ be the function defined over $\set{a}^*$ by 
$$f(a^n)=\left\{
\begin{array}{ll}
n & \text{ if }n\text{ even}\\
\infty & \text{ otherwise}
\end{array}
\right.$$

The minimal stabilization semigroup of $f^\approx$ is : $\set{a,aa,(aa)^\sharp,(aa)^\sharp a}$, with $aa\cdot a=a$ and $(aa)^\sharp a\cdot a=(aa)^\sharp$. Hence the \se~ $aaa(aa)^\sharp$ is well-formed for $f^\approx$ but the \se~$a^\sharp$ is not.
\end{exa}

The \se s that are not well-formed have to be removed from the set we want to quotient, in order to get only real elements of the syntactic semigroup. 

\subsection{\ose s}
We have defined the set of \se s that we want to quotient to get the syntactic equivalence of cost functions. However, we saw that some of these \se s may not evaluate properly relatively to the cost function $f^\approx$ we want to study, and therefore does not correspond to an element in the syntactic stabilization semigroup of $f^\approx$.
\smallskip

Thus we need to be careful about the stabilization operator, and apply it only to ``idempotent \se s''. To reach this goal, we will add an ``idempotent operator'' $\omega$ on \se s, which will always associate an idempotent element (relative to $f^\approx$) to a \se, so that we can later apply $\sharp$ and be sure of creating well-formed expressions for $f$.

We define the set $\oexpr$ of \ose s on an alphabet $\A$ :
$$E:=a\in\A~|~EE~|~E^\omega~|~E^{\omega\sharp}$$

The intuition behind operator $\omega$ is that $x^\omega$ is the idempotent obtained by iterating $x$ (which always exists in finite semigroups).
% If it is properly defined, $a^{\omega\sharp}$ of Example \ref{pair} will be equivalent to $(aa)^{\omega\sharp}$

A \intro{context} $C[x]$ is a \ose~with possible occurrences of a free variable $x$.
Let $E$ be a \ose, $C[E]$ is the \ose~obtained by replacing all occurrences of $x$ by $E$ in $C[x]$, i.e. $C[E]=C[x][x\leftarrow E]$. Let $\coe$ be the set of contexts on \ose s.

We will now formally define the semantic of operator $\omega$, and use \ose s to get a syntactic equivalence on cost functions, without mistyped \se s.
\begin{defi}
If $E\in\oexpr$ and $k,n\in\N$, we define $E(k,n)$ to be the word $E[\omega\leftarrow k,\sharp\leftarrow n]$, where exponentiation is relative to concatenation of words.
\end{defi}

\begin{lem}\label{KF}
Let $F=f^\approx$ be a regular cost function, there exists $K_F\in\N$ such that for any $E\in\oexpr$, the \se~$E[\omega\leftarrow K_F!]$ is well-formed for $F$, and we are in one of these two cases 
\begin{enumerate}
\item $\forall k\geq K_F, \set{f(E(k!,n)),n\in\N}$ is bounded : we say that $E\in F^B$.
\item $\forall k\geq K_F, \lim_{n\rightarrow\infty}f(E(k!,n))=\infty$ : we say that $E\in F^\infty$.
\end{enumerate}
\end{lem}

\begin{proof}
Let $F={f^\approx}$ be a regular cost function recognized by $\semi_F,h,I$.
Let $N=|\semi_F|$. It suffices to take $K_F\geq N$ to verify that for any $E\in\oexpr$, the \se~$E[\omega\leftarrow K_F!]$ is well-formed for $F$. Moreover, if $s\in\semi_F$,  $s^{k!}=s^{K_F!}$ for all $k\geq K_F$.
Let us show that $F^\infty\uplus F^B=\oexpr$.
Let $E\in\oexpr$, and $k\geq K_F$. 
Let $e=E[\omega\leftarrow k!]$, $e$ is well-formed for $\semi_F$.
For all $n\in\N$, let $u_n=e(n)=E(k!,n)$.
The structure of $e$ directly gives us a factorization tree for $u_n$, the height of this tree depending only on $e$.
Thus we know that there exists $\alpha$ (depending on $e$) such that $\rho(h(u_n))\sim_\alpha\eval(e)|_n\eval(u_n)$.

Therefore, $$\eval(e)\in I\Rightarrow \forall n,I[\rho(h(u_n))]\geq_\alpha n\Rightarrow\forall n, f(u_n)\geq_\alpha n\Rightarrow \lim f(u_n)=\infty$$
and  $\eval(e)\notin I\Rightarrow\forall n, I[\rho(h(u_n))]\leq \alpha(1) \Rightarrow\forall n, f(u_n)\leq \alpha(1) \Rightarrow E\in F^B$.
We get that $F^\infty=\set{E\in\oexpr,\eval(E)\in I}$ and $F^B=\set{E\in\oexpr,\eval(E)\notin I}$ which shows the result.

\end{proof}

Here, $F^B$ and $F^\infty$ are the analogs for regular cost functions of ``being in $L$'' and ``not being in $L$'' in language theory. But this notion is now asymptotic, since we look at boundedness properties of quantitative information on words. Moreover, $F^\infty$ and $F^B$ are only defined here for regular cost functions, since $K_F$ might not exist if $f$ is not regular.

\begin{defi}\label{Rf}Let $F$ be a regular cost function, 
we write $E\iffb_F E'$ if $(E\in F^B\Leftrightarrow E'\in F^B)$.
%We define $R_F(E)$ to be the \se~$E[\omega\leftarrow K_F!]$
Finally we define
$$E\equiv_F E'\text{ iff }\forall C[x]\in\coe, C[E]\iffb_F C[E']$$
\end{defi}

\begin{rem}
If $u,v\in\A^*$, and $L$ is a regular language, then $u\sim_L v$ iff $u\equiv_{\chi_L}v$ ( $\sim_L$ being the syntactic congruence of $L$). In this sense, $\equiv$ is an extension of the classic syntactic congruence on languages.
\end{rem}

Now that we have properly defined the equivalence $\equiv_F$ over $\oexpr$, it remains to verify that it is indeed a good syntactic congruence, i.e. $\oexpr/{\equiv_F}$ is the syntactic stabilization semigroup of $F$.

\subsubsection{Structure of $\oexpr/{\equiv_F}$}
If $F$ is a regular cost function, let $\semi_F=\oexpr/{\equiv_F}$. We show that we can provide $\semi_F$ with a structure of stabilization semigroup $\perm{\semi_F,\cdot,\leq,\sharp}$.

If $E\in\oexpr$, let $\cl{E}$ be its equivalence class for the $\equiv_F$ relationship.
We first naturally define the stabilization semigroup operators :
$\cl{E}\cdot\cl{E'}=\cl{EE'}$ and if $\cl{E}$ idempotent we have $\cl{E}=\cl{E^\omega}$ and $(\cl{E})^\sharp=\cl{E^{\omega\sharp}}$. $\leq$ is the minimal partial order induced by the inequalities $s^\sharp\leq s$ where $s$ is idempotent, and compatible with the stabilization semigroup structure.

Let us show that these operations are well-defined :
\begin{itemize}
\item[Product] If $E_1\equiv_F E_1'$ and $E_2\equiv_F E_2'$.
By Lemma \ref{contsim} with context $xE_2$ and $E_1'x$, $E_1E_2\equiv_F E_1'E_2\equiv_FE_1'E_2'$, so $\cl{E_1E_2}=\cl{E_1'E_2'}$.
\item[Stabilization]
If $E\equiv_F E'$, by Lemma \ref{contsim} with context $x^{\omega\sharp}$, $E^{\omega\sharp}\equiv_F E'^{\omega\sharp}$, hence $\cl{E^{\omega\sharp}}=\cl{E'^{\omega\sharp}}$.
\end{itemize}

Moreover, it is easy to check that all axioms of stabilization semigroups are verified, for example $(s^\sharp)^\sharp=s^\sharp$ because for any sequence $u_n$ which is either bounded or tends towards $\infty$, $u_{n^2}$ has same nature as $u_n$.

\begin{thm}\label{congmin}
$\semi_F=\oexpr/{\equiv_F}$ is the minimal stabilization semigroup recognizing $f$.
\end{thm}

\begin{proof}
Let $I_F=\set{\cl{E}, E\in F^\infty}$, and $h_F:\A^*\rightarrow \semi_F^*$ the length-preserving morphism defined by $h_F(a)=\cl{a}$ for all $a\in\A$ (a letter is a particular \ose).

Let $\Smin,h,I$ be the minimal stabilization semigroup recognizing $F$, as defined in appendix A.7 of \cite{CKL}. Let $\rho$ be its compatible mapping, and $\eval:\oexpr\rightarrow\Smin$ the corresponding evaluation function.
We will show that $E\equiv_F E'$ iff $\eval(E)=\eval(E')$.

We know by the proof of Lemma \ref{KF} that $E\in F^\infty\Leftrightarrow\eval(E)\in I$.
We remind that the definition of $\Smin$ is based on the fact that if two elements behave the same relatively to $I$ in any context, they are the same.
These facts give us the following sequence of equivalences :

$\begin{array}{ll}
E\equiv_F E' &\Leftrightarrow \forall C[x]\in\coe, C[E]\iffb_F C[E']\\
&\Leftrightarrow \forall C[x]\in\coe, (C[E]\in F^\infty\Leftrightarrow C[E']\in F^\infty)\\
&\Leftrightarrow  \forall C[x]\in\coe, (\eval(C[E])\in I\Leftrightarrow(\eval(C[E'])\in I)\\
&\Leftrightarrow \eval(E)=\eval(E')\\
\end{array}$

This gives a bijection between $\semi_F$ and $\Smin$ ($\eval$ function is surjective on $\Smin$, by minimality of $\Smin$). Moreover, this bijection is an isomorphism, since in both semigroups, operations are induced by concatenation and $\sharp$ on \se s.
$h$ is determined by its image on letters, so we have to define $h_F(a)=\cl{a}$ to remain coherent.
Finally, we have $\eval(E)\in I\Leftrightarrow E\in F^\infty$,
therefore the set $I_F$ corresponding to $I$ in the bijection is $I_F=\set{\cl{E},E\in F^\infty}$.

\end{proof}

\subsection{Details on \ose s}

\begin{lem}\label{contsim}
If $E\equiv_F E'$, then for any context $C_1[x]\in\coe$, $C_1[E]\equiv_F C_1[E']$.
\end{lem}
\begin{proof}
Let $E,E'$ and $C_1[x]$ defined by the Lemma.
Let $C[x]$ be a context. We define $C'[x]=C[C_1[x]]$.
The definition of the $\equiv_F$ relation implies $C'[E]\iffb_F C'[e']$. Hence $C[C_1[e]]\iffb_F C[C_1[E']]$.

This is true for any context $C[x]$ so $C_1[E]\equiv_F C_1[E']$.

\end{proof}
\begin{prop}
The relation $\equiv_F$  does not change if we restrict contexts to having only one occurence of $x$, as it was done for $\expr$ in \cite{CKL}.
\end{prop}

\begin{proof}
Let $\equiv'_F$ be the equivalence relation defined with single-variable contexts. we just need to show that $E\equiv'_F E'\implies E\equiv_F E'$ (the converse is trivial).
Let us assume $E\equiv'_F E'$, and let $C[x_1,x_2]$ be a context with two occurences of $x$, labelled $x_1$ and $x_2$.
Then $C[E]=C[x_1\leftarrow E,x_2\leftarrow E]\iffb_F C[x_1\leftarrow E,x_2\leftarrow E']\iffb_F C[x_1\leftarrow E',x_2\leftarrow E']=C[e']$. The generalization to an arbitrary number of occurences of $x$ is obvious, and we get $E\equiv_F E'$.

\end{proof}
\bigskip

\subsubsection{Growing speeds lemma}
The following lemma will be used for technical purposes in future proofs. We state it here because it is an intuitive statement which can give a better understanding of the behaviour of regular cost functions and \se s.

\begin{lem}\label{sharpfun}
Let $F=f^\approx$ be a regular cost function, and $e\in\expr$ containing $N$ $\sharp$-operators $\sharp_1,\dots,\sharp_N$. For all $i\in\set{1,\dots,N}$, let $\sigma_i$ be a function $\N\rightarrow\N$ with $\sigma_i(n)\rightarrow\infty$.
Then $$f(e[\sharp_i\leftarrow\sigma_i(n)\text{ for all }i])\rightarrow\infty\Leftrightarrow f(e(n))\rightarrow\infty.$$
In other words, we can replace some $n$ exponents by any function $\sigma(n)\rightarrow\infty$ when approximating a \se~by a sequence of words. It does not change the nature of the sequence relatively to $f$.
\end{lem}

\begin{proof}

This result is intuitive : since we always work up to cost equivalence, growing at different speeds has an effect on correction functions, but not on qualitative behaviour.

We will use notation $\bowtien$ : $g_1(n)\bowtien g_2(n)$ means ``$g_1(n)$ is bounded iff $g_2(n)$ is bounded''. Remark that here all functions will either be bounded or tend towards $\infty$, thanks to the constraint that \se s are well-formed for $\semi_F$. 
 %
%For convenience we will note $e_n=e[\sharp_i\leftarrow\sigma_i(n)\text{ for all }i]$. We want to show that $f(e_n)\bowtien f(e(n))$. 
%On introduit une notation qui permettra de simplifier les explications.
%Si $a_n$ et $b_n$ sont deux expressions qui contiennent la variable libre $n$, et à valeur dans $\N$, on dira que $a_n\bowtien b_n$ si $\limsup a_n=\infty \Leftrightarrow \limsup b_n=\infty$.
%
%On peut remarquer qu'ici, toutes les suites seront soit bornées, soit tendant vers l'infini, gr\^ace à la contrainte stipulant que $e$ est bien formée pour $\semi_F$. Par exemple $e$ peut \^etre obtenue à partir d'une \ose, en remplaçant $\omega$ par $K_F!$.

We will note $e_n=e[\sharp_i\leftarrow\sigma_i(n)\text{ pour tout }i]$. We want to show that $f(e_n)\bowtien f(e(n))$.
Let $\rho$ be compatible with $\semi_F$.

We want to show that there is $\alpha$ such that for all $n$, $\rho(e_n)\sim_\alpha \rho(e(n))$.
We proceed by induction on $N$.
If $N=0$, then $e_n=e(n)$ and the result is trivial.

We suppose the result true for $k<N$, with function $\alpha_<$. Let $\sharp_N$ be an outmost stabilization operator (i.e. not nested in an other $\sharp$).
We can write $e=rs^{\sharp_N}t$, with $r,s,t\in\expr$, well-formed for $\semi_F$, and $\eval(s)\in E(\semi_F)$.
%
%Soit $\beta$ donné par le Théorème \ref{compaxioms}, et $\gamma$ tel que $n\sim_\gamma\sigma_N(n)$.
%On a
%\begin{align*}
%\rho(e_n)&=\rho(r_n(s_n)^{\sigma_N(n)}t_n)\\
%&\sim_{\beta}\tilde\rho(\rho(r_n)\rho((s_n)^{\sigma_N(n)})\rho(t_n))\\
%&\sim_{\gamma}\tilde\rho(\rho(r_n)\rho((s_n)^n)\rho(t_n))\\
%&\sim_{\alpha_<}\tilde\rho(\rho(r(n))\rho(s(n)^n)\rho(t(n)))\\ 
%&\sim_{\beta}\rho(e(n)).
%\end{align*}

By induction, there are $n$-trees of bounded height, and of value $\rho(r(n))$, $\rho(s(n))$ and $\rho(t(n))$ over $r_n$, $s_n$ and $t_n$ respectively.
We can combine these trees by two binary nodes, and by a node which is either idempotent of stabilizing, in the following way:te :

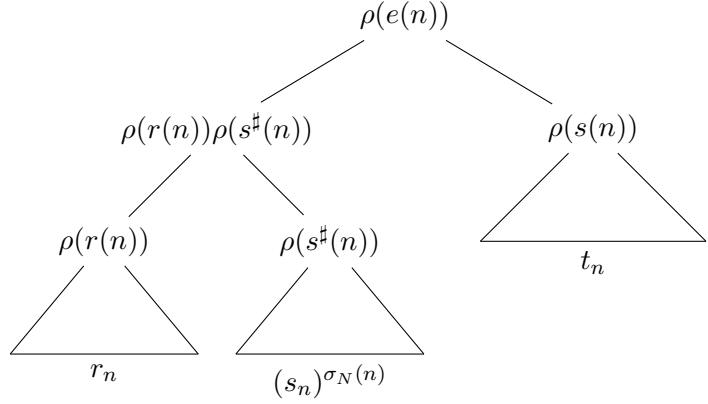
\begin{center}
%\begin{figure}[h]\caption{Le $n$-sous-calcul $t_u$}
\begin{tikzpicture}[ level 1/.style={sibling distance=5cm},
 level 2/.style={sibling distance=3cm},
  level 3/.style={sibling distance=2.5cm}]
\node{$\rho(e(n))$}
child {node {$\rho(r(n))\rho(s^\sharp(n))$}
	child {node {$\rho(r(n))$} child  {coordinate (a1)} child {coordinate (a2)}}
	child {node{$\rho(s^\sharp(n))$} child {coordinate (b1)} child {coordinate (b2)}}
}
child {node {$\rho(s(n))$} child {coordinate (c1)} child {coordinate (c2)}}
;
\draw (a1) -- (a2) node[midway,below]{$r_n$};
\draw (b1) -- (b2) node[midway,below]{$(s_n)^{\sigma_N(n)}$};
\draw (c1) -- (c2) node[midway,below]{$t_n$};
\end{tikzpicture}
%\end{figure}
\end{center}
The tree that we obtain can use sometimes $n$, sometimes $\sigma_N(n)$ as a threshold. It will be either an over-approximation or an under-approximation of the value of $f(n)$, with an error controlled by $\sigma_N$.
Thus the sequence of values generated at the root is $\sim$-equivalent to $\rho(e(n))$, wihle the word $e_n$ is always the leaf words. This concludes the proof of $\rho(e_n)\sim\rho(e(n))$
%Ceci nous permet de construire une fonction $\alpha$ vérifiant la propriété voulue, à l'aide de la Propriété \ref{transcost}.
\end{proof}
\subsection{Case of unregular cost functions}

The syntactic congruence can still be defined on unregular languages, and the number of equivalence classes becomes infinite, whereas a priori, we need cost functions to be regular to define their syntactic congruence.

Here, if $F=f^\approx$ is not regular, $\equiv_F$ may not be properly defined, since we use the existence of the minimal stabilization semigroup of $F$ to give a semantic to the operator $\omega$.
But we can go back to \se s and define $\sim_F$ on $\expr$ for all $f$ in the following way : $e\sim_F e'$ if for any context $C[x]$ on \se s, the set $\set{f(C[e])(n), n\in\N}$ is bounded iff $\set{f(C[e'])(n), n\in\N}$ is bounded.

In this way if $F$ is regular, then for all $e,e'\in\expr$, $e\sim_F e'$ iff $e[\sharp\leftarrow\omega\sharp]\equiv_F e'[\sharp\leftarrow\omega\sharp]$.
In particular $\expr/{\sim_F}$ is bigger than $\oexpr/{\equiv_F}$ when $f$ is regular : there might be equivalence classes corresponding to \se s that are not well-formed for $F$.

However, if $F$ is not regular, $\expr/{\sim_F}$ is not infinite in general (this differs from the results in language theory).

\begin{exa}
Let $f(u)=\min_{e\in\expr}\set{|e|, \exists n\in\N, u=e(n)}$, and $F=f^\approx$, there is only one equivalence class for $\sim_F$, because $f(C[e](n))$ is always bounded by $|C[e]|$. So $\expr/{\sim_F}$ has only one element, and therefore cannot contain a stabilization semigroup computing $F$. This gives us a proof that $F$ is not regular.
\end{exa}

\section{Expressive power of $\ltlq$}

If $F$ is a regular cost function, we will call $\semi_F$ the syntactic stabilization semigroup of $F$.

A finite semigroup $\semi=\perm{S,\cdot}$ is called \intro{aperiodic} if $\exists k\in\N,\forall s\in\semi,s^{k+1}=s^k$. The definition is the same if $\semi$ is a finite stabilization semigroup.

%\begin{lemme}
%%a definir, lemme d'extraction de suite
%\end{lemme}
\begin{rem}
For a regular cost function $F$, the statements ``$F$ is recognized by an aperiodic stabilization semigroup'' and ``$\semi_F$ is aperiodic'' are equivalent, since $\semi_F$ is a quotient of all stabilization semigroups recognizing $F$.
\end{rem}

\subsection{From $\ltlq$ to Aperiodic Stabilization Semigroups}

\begin{thm}\label{ltlap}
Let $F$ be a cost function described by a $\ltlq$-formula, then $F$ is regular and the syntactic stabilization semigroup of $F$ is aperiodic.
\end{thm}
The proof of this theorem will be the first framework to use the syntactic congruence on cost functions.

%We will use \ose s to characterize aperiodic stabilization semigroup : $\semi_F$ is aperiodic iff there exists $k\in\N$ such that for any \ose~$E$, $E^k\equiv_fE^{k+1}$.

%If $\phi$ is a $\ltlq$-formula, we will say that $\phi$ verifies property $AP$ if there exists $k\in\N$ such that for any \ose~$E$, $E^k\equiv_{\sem{\phi}}E^{k+1}$, which is equivalent to ``$\sem{\phi}$ has an aperiodic syntactic stabilization semigroup''.
%
%With this in mind, we can do an induction on $\ltlq$-formulaes : 
%we first show that $\semi_\Omega$ and all $\semi_a$ for $a\in\A$ are aperiodic.
%
%We then proceed to the induction on $\phi$ : assuming that $\varphi$ and $\psi$ verify property $AP$, we show that $X\psi$, $\varphi\vee\psi$, $\varphi\wedge\psi$, $\varphi U\psi$ and $\varphi \UN\psi$ verify property $AP$.

\begin{proof}

We want to show that for all $\ltlq$-formula $\phi$, $\ssp$ is aperiodic.

We proceed by an induction on $\phi$ and use the characterization of $\ssp$ provided by Theorem \ref{congmin}.

\subsubsection{Case $\phi=a$}~

We have $S_{\sphapp}=\set{a,b}$ with $a\cdot b=a\cdot a=a$, and $b\cdot a=b\cdot b=b$, it is aperiodic (also trivial if $\phi=\neg a$).

\subsubsection{Case $\phi=\Omega$}~

Then $S_{\sphapp}=\set{1,a}$ with $1$ neutral element and $a\cdot a=a$, it is aperiodic.

\subsubsection{Case $\phi=\varphi_1\wedge\varphi_2$ or $\phi=\varphi_1\vee\varphi_2$}~

$\phi$ is recognized by the product semigroup of $\semi_{\sem{\varphi_1}}$ and $\semi_{\sem{\varphi_2}}$, which is aperiodic by induction hypothesis. 

\subsubsection{Case $\phi=X\psi$}~

We know by induction hypothesis that $\semi_{\sem\psi^\approx}$ is aperiodic, so there exists $k\in\N$ such that for any \ose~$E$, $E^k\equiv_{\sem{\psi}}E^{k+1}$. We want to show that it is also true for $\sph$.
Let $E$ be a \ose, and $e=E[\omega\leftarrow\max(K_{\sph^\approx}!,K_{\sem\psi^\approx}!)]$ (from Lemma \ref{KF}).

We want to show that $E^{k+2}\equiv_{\sphapp}E^{k+1}$ i.e. for any context $C[x],$ 
$$\sph(C[e^{k+2}](n))\bowtien\sph(C[e^{k+1}](n)).$$

Let $C[x]$ be a context.
\begin{itemize}
\item If $C[x]=aC'[x]$, then by proposition \ref{contsim} with context $xe$: $$\sph(C[e^{k+2}](n))=\sem{\psi}(C'[e^{k+2}](n))\bowtien\sem{\psi}(C'[e^{k+1}](n))=\sph(C[e^{k+1}](n)).$$

\item If the beginning of $C[x]$ is a letter $a$ under (at least) a $\sharp$, we have a context $C'[x]$ such that for any \se~$e'$, $C[e'](n+1)=aC'[e'](n)$.
For instance if $C[x]=((ax)^\sharp b)^\sharp$ then $C'[x]=x(ax)^\sharp b((ax)^\sharp b)^\sharp$.
 Then we can write $\sph(C[e^{k+2}](n+1))=\sem{\psi}(C'[e^{k+2}](n))\bowtien\sem{\psi}(C'[e^{k+1}](n))=\sph(C[e^{k+1}](n+1))$.

\item Finally, if  $C[x]$ starts with $x$ (possibly under $\sharp$), we expand $x$ in $ex$ in $C[x]$, so that it does not start with $x$ anymore. As before we can get $C'[x]$ such that $C[e^{k+1}](n+1)=aC'[e^k](n)$ and $C[e^{k+2}](n+1)=aC'[e^{k+1}](n)$ for all $n$, hence 
$$\begin{array}{ll}
\sph(C[e^{k+2}](n+1))&=\sph(aC'[e^{k+1}](n))\\
&=\sem{\psi}(C'[e^{k+1}](n))\\
&\bowtien\sem{\psi}(C'[e^{k}](n))\\
&=\sph(aC'[e^{k}](n))\\
&=\sph(C[e^{k+1}](n+1))\\
\end{array}$$
\end{itemize}

\subsubsection{Case $\phi=\varphi U\psi$}~

we know by induction hypothesis that $\semi_{\sem{\varphi}^\approx}$ and $\semi_{\sem{\psi}^\approx}$ are aperiodic, so there exists $k\in\N$ such that for any \ose~$E$, $E^k\equiv_{\sem{\varphi}}E^{k+1}$ and $E^k\equiv_{\sem{\psi^\approx}}E^{k+1}$. 
Let $E$ be a \ose.
We will show that  $E^{k+1}\equiv_{\sphapp}E^{k+2}$

Let $C[x]$ be a context in $\coe$, $K=\max(K_{\sem{\varphi}^\approx},K_{\sem\psi^\approx})$, $u_n=C[E^{k+1}](K!,n)$ and $v_n=C[E^{k+2}](K!,n)$.
We want to show that $C[E^{k+1}]\iffb_{\sphapp}C[E^{k+2}]$, i.e. $\sph(u_n)\bowtien\sph(v_n)$.
Assume for example that $\sph(u_n)$ is bounded by $m$
We have $u_n,m\models\phi$ for all $n$.
We can write $u_n=y_nz_n$ with $z_n,m\models\psi$ and for any strict suffix $y_n^i$ of $y_n$, $y_n^iz_n,m\models\varphi$. Let $p_n$ be the starting position of $z_n$ (position $0$ being the beginning of the word). We define $y_n^i$ to be the suffix of $y_n$ starting at position $i$ for all $i\in\sem{0,p-1}$. In this way $y_n^0=y_n$.
\begin{center}
\begin{tikzpicture}
%\node at (-2,0){$\varphi \U \psi$:};
%\foreach \x in  {0,1,2,3,4,5,6,7}{
%\node at (.4*\x,-.2){$a_\x$};
%}

\def\l{8}
\def\n{10}
\def\nm{9}

%lines
\draw (0,0) -- (\l,0);
\foreach \y in {0, \l*.7, \l}{
\draw (\y,-.2) -- (\y,.2);
}

%formulas
\foreach \x in  {1,...,\nm}{
\node at (\l*.7*\x/\n,.5){$\varphi$};
}
\node at (\l*.7,.5){$\psi$};

%words
\node at (-.5,0){$u_n:$};
\node at (\l*.7/2,-.5){$y_n$};
\node at (\l*.7,-.4){$p_n$};
\node at (\l*.85 ,-.5){$z_n$};
\end{tikzpicture}
\end{center}

Let us focus on the position $p_n$ of the beginning on $z_n$. The \se~$e=C[E^{k+1}](K!)$ is finite so we can extract a sequence $u_{\delta(n)}$ from $u_n$ such that the beginning position $p_{\delta(n)}$ of $z_{\delta(n)}$ corresponds to the same position $p$ in $e$.
Let $\set{e_j,j\in J}$ be the finite set of $\sharp$-expression such that $e_j^\sharp$ contains position $p$ in $e$.
We choose $J=\set{1,r}$ with $1\leq j<j'\leq r$ implies $e_j^\sharp$ is a subexpression of $e_{j'}$.
For convenience, we label the $\sharp$-operator of $e_j^\sharp$ with $j$.
%Moreover, we define the partition of $J=J_C\cup J_E$, with $j\in J_C$ if the $\sharp_j$ is in $C[x]$, and $j\in J_E$ if it is in $E$.
Note that $J$ can be empty, if $p$ does not occur under a $\sharp$ in $e$.

We denote by $\overleftarrow{f_j}({\delta(n)})$ the number of occurences of $e_j({\delta(n)})$ (coming from the corresponding $e_j^\sharp$) in $y_{\delta(n)}$ and  we define $\overrightarrow{f_j}({\delta(n)})$ in the same way relatively to $z_{\delta(n)}$. We have for all $n\in\N$, ${\delta(n)}-1\leq\overleftarrow{f_j}({\delta(n)})+\overrightarrow{f_j}({\delta(n)})\leq {\delta(n)}$. The ${\delta(n)}-1$ lower bound is due to the fact than $p$ can be in the middle of one occurence of $e_j$, therefore this occurence does not appear in $y_{\delta(n)}$ nor in $z_{\delta(n)}$.

This implies that for each $j\in J$, we are in one of these three cases: 
\begin{itemize}
\item $j\in J_1$: $\overleftarrow{f_j}(\delta(n))$ is unbounded and $\overrightarrow{f_j}(\delta(n))$ is bounded.
\item $j\in J_2$: $\overleftarrow{f_j}(\delta(n))$ is bounded and $\overrightarrow{f_j}(\delta(n))$ is unbounded.
\item $j\in J_3$: $\overleftarrow{f_j}(\delta(n))$ and $\overrightarrow{f_j}(\delta(n))$ are both unbounded .
\end{itemize}

But $J$ is finite, hence we can extract $\sigma(n)$ from $\delta(n)$ such that for each $j\in J$: 
\begin{itemize}
\item If $j\in J_1$, $\overleftarrow{f_j}(\sigma(n))\rightarrow\infty$ and $\overrightarrow{f_j}(\sigma(n))$ is constant.
\item If $j\in J_2$, $\overleftarrow{f_j}(\sigma(n))$ is constant and $\overrightarrow{f_j}(\sigma(n))\rightarrow\infty$.
\item If $j\in J_3$, $\overleftarrow{f_j}(\sigma(n))\rightarrow\infty$ and $\overrightarrow{f_j}(\sigma(n))\rightarrow\infty$.
\end{itemize}
Remark that if $j<j'$ and $\overrightarrow{f_j}\circ\sigma\neq 0$, then $j\notin J_1$.
Symmetrically, if $j<j'$ and $\overleftarrow{f_j}\circ\sigma\neq 0$, then $j\notin J_2$.

We can distinguish three cases for the position of $p$ in $e=C[E^{k+1}](K!)$:
\bigskip

\textbf{First case}: $p$ is before the first occurence of $E$ in $e$.
\begin{center}
\begin{tikzpicture}
%\node at (-2,0){$\varphi \U \psi$:};
%\foreach \x in  {0,1,2,3,4,5,6,7}{
%\node at (.4*\x,-.2){$a_\x$};
%}

\def\l{10}
\def\n{8}

%lines
\draw (0,0) -- (\l,0);
\foreach \y in {0,\l}{
\draw (\y,-.2) -- (\y,.2);
}

\foreach \x in {1,...,\n}{
\pgfmathparse{\l*(\x+2)/(\n+5)}
\let\y\pgfmathresult
\draw (\y,-.2) -- (\y,.2);
\ifthenelse{\x<\n}
{
\pgfmathparse{\l*(\x+2.5)/(\n+5)}
\let\z\pgfmathresult
\node at (\z,-.3){$E$};
}
{}
}

%position p
\pgfmathparse{\l*.1}
\let\y\pgfmathresult
\draw (\y,-.2) -- (\y,.2);
\node at (\y,-.4){$p$};

%words
\node at (-.5,0){$e:$};

\end{tikzpicture}
\end{center}
We consider $C'[x]\in\coe$ obtained from $C[x]$ by replacing $\sharp_j$ by the constant value of $\overrightarrow{f_j}(\sigma(n))$ for all $j\in J_1$.
We have $\sem\psi(z_{\sigma(n)})\leq m$ for all $n$, but by Lemma \ref{sharpfun}, $\sem\psi(z_n)$ is bounded iff $C'[E^{k+1}]\in(\sem\psi^\approx)^B$. By induction hypothesis, $C'[E^{k+1}]\in\sem\psi^B\Leftrightarrow C'[E^{k+2}]]\in(\sem\psi^\approx)^B$.
Let $z'_n$ be the suffix of $C[E^{k+2}](K!,n)$ starting at position $p_n$.
By reusing Lemma \ref{sharpfun}, we get that $\sem\psi(z'_{\sigma(n)})\leq m'$ for some $m'$.

We still have to show that there exists a constant $M$ such that $\sem{\varphi}(y^i_{\sigma(n)}z'_{\sigma(n)})\leq M$ for all $n$ and all $i\in\sem{1,p_{\sigma(n)}}$ (the $y^i_{\sigma(n)}$ are not affected by the change from $E^{k+1}$ to $E^{k+2}$).
Let us call $g^i_{\sigma(n)}=\sem{\varphi}(y^i_{\sigma(n)}z'_{\sigma(n)})$ for more lisibility.
Let us assume that no such $M$ exists, then $\set{g^i_{\sigma(n)}:n\in\N, 1\leq i\leq p_{\sigma(n)}}$ is unbounded.
For all $n$, we define $i_n$ such that $g^{i_{\sigma(n)}}_{\sigma(n)}=\max\set{g^i_{\sigma(n)}: 1\leq i\leq p_{\sigma(n)}}$.
By construction, the sequence $g^{i_{\sigma(n)}}_{\sigma(n)}=\sem{\varphi}(y^{i_{\sigma(n)}}_{\sigma(n)}z'_{\sigma(n)})$ is unbounded. We first extract $\sigma'(n)$ from $\sigma(n)$ such that $g^{i_{\sigma'(n)}}_{\sigma'(n)}\to\infty$.

We can now repeat the same process as before to extract a sequence $\gamma(n)$ from $\sigma'(n)$, such that the starting positions of $y^{i_{\gamma(n)}}_{\gamma(n)}$ for all $n$ correspond to the same position in $e$, and such that there exists a context $C''[x]$ with $\sem{\varphi}(y^{i_{\gamma(n)}}_{\gamma(n)}z_{\gamma(n)})\bowtien\sem{\varphi}(C''[E^{k+1}](K!,\gamma(n)))$ (by Lemma \ref{sharpfun} again).
By adding an extra $E$ (from $k+1$ to $k+2$) and changing $z$ by $z'$ (the $y$ factors are not concerned by occurences of $E$), we get $g^{i_{\gamma(n)}}_{\gamma(n)}\bowtien\sem{\varphi}(C''[E^{k+2}](K!,\gamma(n)))$.
By hypothesis, $\sem{\varphi}(y^{i_{\gamma(n)}}_{\gamma(n)}z_{\gamma(n)})$ bounded by $m$, and  $C''[E^{k+1}]\iffb_{\sem{\varphi}}C''[E^{k+2}]$, so $g^{i_{\gamma(n)}}_{\gamma(n)}$ is bounded, but we already know that $g^{i_{\gamma(n)}}_{\gamma(n)}\to\infty$. We have a contradiction, so $M$ must exist.

We finally obtain the existence of $M$ such that for all $n$ and valid $i$,  $\sem{\varphi}(y^i_{\sigma(n)}z'_{\sigma(n)})\leq M$. This together with the previous result on $\psi$ gives us that $\sem{\varphi U\psi}(C[E^{k+2}](K!,n))\leq\max(m',M)$. We got $C[E^{k+1}]\in\sph^B\implies C[E^{k+2}]\in{\sphapp}^B$. The other direction works exactly the same, by removing one $E$ instead of adding one. Hence we have $C[E^{k+1}]\iffb_{\sphapp}C[E^{k+2}]$.
\bigskip

\textbf{Second case}: $p$ is after the last occurence of $E$ in $e$.
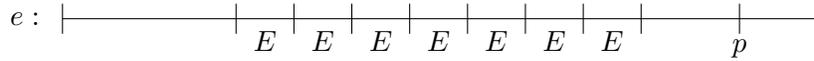
\begin{center}
\begin{tikzpicture}

\def\l{10}
\def\n{8}

%lines
\draw (0,0) -- (\l,0);
\foreach \y in {0,\l}{
\draw (\y,-.2) -- (\y,.2);
}

\foreach \x in {1,...,\n}{
\pgfmathparse{\l*(\x+2)/(\n+5)}
\let\y\pgfmathresult
\draw (\y,-.2) -- (\y,.2);
\ifthenelse{\x<\n}
{
\pgfmathparse{\l*(\x+2.5)/(\n+5)}
\let\z\pgfmathresult
\node at (\z,-.3){$E$};
}
{}
}

%position p
\pgfmathparse{\l*.9}
\let\y\pgfmathresult
\draw (\y,-.2) -- (\y,.2);
\node at (\y,-.4){$p$};

%words
\node at (-.5,0){$e:$};

\end{tikzpicture}
\end{center}
This time $z_n$ is not affected by changing from $E^{k+1}$ to $E^{k+2}$, however it affects some of the $y^i_n$.
Let $y'^i_nz_n$ be the suffixes of $v_n=C[E^{k+2}](K!,n)$, and $p'_n$ the position of the beginning of $z_n$ in $v_n$.

As before, we assume that $\set{\sem{\varphi}(y'^i_{\sigma(n)}z_{\sigma(n)}) : n\in\N, 1\leq i\leq p_{\sigma(n)}}$ is unbounded, and we build a sequence $y'^{i_{\gamma(n)}}_{\gamma(n)}$ with the same start position in $e$, such that $\sem{\varphi}(y'^{i_{\gamma(n)}}_{\gamma(n)}z_{\gamma(n)})\to\infty$.

We can again extract context $C''[x]$, but we may need to use again Lemma \ref{sharpfun}, in order to map the $\sharp$'s of $C''[x]$ with the remaining repetitions of idempotent elements, (which could be any functions $g(n)<n$). The main idea is to map positions in $v_{\gamma(n)}$ with positions in $u_ {\gamma(n)}$ in order to be able to bound the values $\sem{\varphi}(y'^{i_{\gamma(n)}}_{\gamma(n)}z_{\gamma(n)})$ with what we know about the behaviour on $u_{\gamma(n)}$, and so get a contradiction.
Three cases are to be distinguished:
\begin{itemize}
\item If a factor corresponding to $E^{k+2}$ occurs in the $y'^{i_{\gamma(n)}}_{\gamma(n)}$, the precedent proof stays valid, and we can map $y'^{i_{\gamma(n)}}_{\gamma(n)}$ with some $y^{j_{\gamma(n)}}_{\gamma(n)}$ ($j_{\gamma(n)}$ may be different from $i_{\gamma(n)}$) in order to get the contradiction. The mapping just need to take in account the shift due to the new occurences of $E$, but the positions in the words are essentially the sames.
\item If the remaining factors contain at most $k$ occurences of $E$, then the position can be matched with positions in $u_n$ without any changes, and we get the contradiction.
\item If the remaining factors contain $k+1$ occurences of $E$, then we can use the equivalence $E^{k+1}\equiv_{\sem{\varphi}^\approx} E^k$ to match positions in $v_n$ with positions in $u_n$ and get the contradiction. This time we map positions in the first $E$ of each sequence $E^k$ with the corresponding position in the second one. Informally, we ``duplicate'' the first $E$ of each sequence.
\end{itemize}

To sum up, the following figure shows how positions of $v_n$ are mapped with positions of $u_n$, in the case $k=2$. This figure is just an example, and is simpler than the general case, because only one sequence of $E$'s appear here.
If an other sequence appears before, all the positions are shifted, but the general principle stays the same.
\begin{center}
\begin{figure}[h]
\begin{tikzpicture}

\def\l{10}
\def\n{8}

%lines
\draw (0,0) -- (\l,0);
\foreach \y in {0,\l}{
\draw (\y,-.2) -- (\y,.2);
}

\foreach \x in {1,...,\n}{
\pgfmathparse{(\x+2)*.8}
\let\y\pgfmathresult
\draw (\y,-.2) -- (\y,.2);
\ifthenelse{\x<\n}
{
\pgfmathparse{(\x+2.5)*.8}
\let\z\pgfmathresult
\node at (\z,.3){$E$};
}
{}
}

%position p
\pgfmathparse{9}
\let\y\pgfmathresult
\draw (\y,-.2) -- (\y,.2);
\node at (\y,.4){$p$};

%words
\node at (-.5,0){$u_n:$};

\begin{scope}[yshift=-1.5cm]

\def\l{10.8}
\def\n{9}

%lines
\draw (0,0) -- (\l,0);
\foreach \y in {0,\l}{
\draw (\y,-.2) -- (\y,.2);
}

\foreach \x in {1,...,\n}{
\pgfmathparse{(\x+2)*.8}
\let\y\pgfmathresult
\draw (\y,-.2) -- (\y,.2);
\ifthenelse{\x<\n}
{
\pgfmathparse{(\x+2.5)*.8}
\let\z\pgfmathresult
\node at (\z,-.3){$E$};
}
{}
}

%position p
\pgfmathparse{9.8}
\let\y\pgfmathresult
\draw (\y,-.2) -- (\y,.2);
\node at (\y,-.4){$p$};

%words
\node at (-.5,0){$v_n:$};
\end{scope}

%traits
\foreach \x in {0,...,35}{
\pgfmathparse{\x*6.4/35}
\let\y\pgfmathresult
\draw[dashed] (\y,0) -- (\y,-1.5);
}

%\foreach \x in {0,...,5}{
%\pgfmathparse{5.6+\x*.8/5}
%\let\y\pgfmathresult
%\draw[dashed] (\y,0) -- (\y,-1.5);
%\draw[dashed] (\y,0) -- (\y+.8,-1.5);
%}

\foreach \x in {0,...,20}{
\pgfmathparse{5.6+\x*(9-5.6)/20}
\let\y\pgfmathresult
\draw[dashed] (\y,0) -- (\y+.8,-1.5);
}
\end{tikzpicture}
\caption{Association of positions}\label{fig:assoc}
\end{figure}
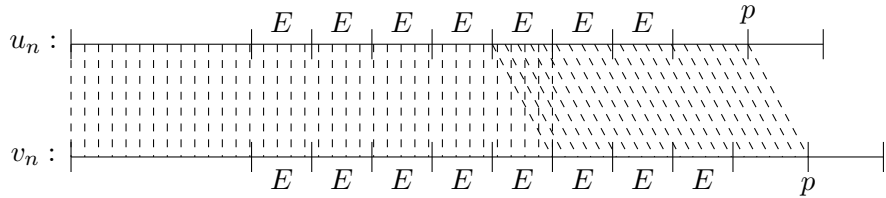
\end{center}
This choice of association is arbitrary: one can indeed choose any $E$ to duplicate, we will still be able to use the induction hypothesis on $u_n$, or the hypothesis that $\varphi$ is true on suffixes of $u_n$ starting before $p$, in order to conclude.
\bigskip

\textbf{Third case}:
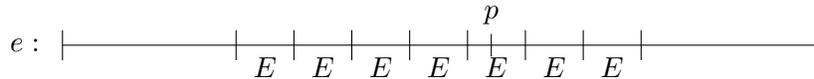
\begin{center}
\begin{tikzpicture}

\def\l{10}
\def\n{8}

%lines
\draw (0,0) -- (\l,0);
\foreach \y in {0,\l}{
\draw (\y,-.2) -- (\y,.2);
}

\foreach \x in {1,...,\n}{
\pgfmathparse{\l*(\x+2)/(\n+5)}
\let\y\pgfmathresult
\draw (\y,-.2) -- (\y,.2);
\ifthenelse{\x<\n}
{
\pgfmathparse{\l*(\x+2.5)/(\n+5)}
\let\z\pgfmathresult
\node at (\z,-.3){$E$};
}
{}
}

%position p
\pgfmathparse{\l*.57}
\let\y\pgfmathresult
\draw (\y,-.15) -- (\y,.15);
\node at (\y,.4){$p$};

%words
\node at (-.5,0){$e:$};

\end{tikzpicture}
\end{center}
 In all other situations, a combination of the techniques used above gives us the wanted result. We just need to do with $\psi$ what we did with $\varphi$ in the second case: for instance we may use $E^{k+1}\equiv_{\sem\psi^\approx}E^{k}$ if $z'_{\sigma(n)}$ contains $k+1$ occurences of $E$.

As before, the other way is similar, and we finally get $E^{k+1}\equiv_{\sphapp}E^{k+2}$.
In conclusion, $\ssp$ is aperiodic.

\subsubsection{Case $\phi=\varphi\UN\psi$}~

We just need to adapt the precedent proof to take in account some exceptions in the validities of $\varphi$ formulae. Indeed removing an occurence of $E$ does not change the number of possible mistakes, but adding one can double it (at worse), since at most two positions in $v_n$ are mapped to the same position in $u_n$.
Hence , under the hypotheses $E^{k}\equiv_{\sem{\psi}^\approx}E^{k+1}$ and  $E^{k}\equiv_{\sem{\varphi}^\approx}E^{k+1}$, we get $E^{k+1}\equiv_{\sem{\varphi\UN\psi}^\approx}E^{k+2}$, with a correction function that doubles the one in the precedent proof.
We can conclude that $\ssp$ is also aperiodic in this case.
\end{proof}

\subsection{From Aperiodic Stabilization Semigroups to $\ltlq$}

\begin{thm}\label{apltl}
Let $F$ be a cost function recognized by an aperiodic stabilization semigroup, then $F$ can be described by an $\ltlq$-formula.
\end{thm}

%The proof of this theorem is a generalization of the proof of Wilke for aperiodic languages in \cite{Wilke99}. However difficulties inherent to quantitative notions appear here.
%
%The main issue comes from the fact that in the classical setting, computing the value of a word in a monoid returns a single element. This fact is used to do an induction on the size of the monoid, by considering the set of possible results as a smaller monoid. The problem is that with cost functions, there is some additional quantitative information, and we need to associate a sequence of elements of a stabilization monoid to a single word. Therefore, it requires some technical work to come back to a smaller stabilization monoid from these sequences.

%We remind the theorem we want to prove:
%
%Let $F$ be a cost function recognized by an aperiodic stabilization semigroup, then $F$ can be described by an $\ltlq$-formula.
\begin{proof}
This proof is a generalization of the proof from Diekert and Gastin for aperiodic languages in \cite{DiekertG}.

Let us first notice that ``$\semi_F$ is aperiodic'' is equivalent to ``$F$ is computed by an aperiodic stabilization monoid'', since aperiodicity is preserved by quotient and by addition of a neutral element.

We take an alphabet $\A\subseteq\M$ to avoid using a morphism $h$ and simplify the proof. The $\ltlq$-formulae are about elements of $\M$, and are monotonic in the sense that $\sem{a}(bu)=0$ iff $b\geq a$, $\infty$ otherwise. It is easy to get from this to the general case by substituting in the formula an element $m$ by $\vee_{h(a)\geq m} a$.
We also will be sloppy with the empty word $\varepsilon$. It is not more difficult to take it in account, but the addition of a lot of special cases for $\varepsilon$ in the proof would make it harder to follow. 

We will always always assume that stabilisation monoids considered here are equipped with the minimal order $\leq$ compatible with the axioms. This means that the only pairs $(x,y)\in \M^2$ such that $x\leq y$ are the ones induced by the axioms of stabilisation monoids.

We assume that $F=f^\approx$ on alphabet $\A\subseteq\M$ is computed by $\M,\mathit{id},I$ with $\M$ aperiodic. Let $\rho$ be compatible with $\M$.

If $m\in\M$, we note $f_m$ the cost function $f_m(u)=\inf\set{n:\rho(u)(n)\geq m}$.
It is sufficient to show that the $f_m$ functions are $\ltlq$-computable up to $\approx$, since $f\approx\min_{m\notin I}f_m$.

We proceed by induction on both the size of the stabilization monoid and on the size of the alphabet, the induction parameter being 
$(|\M|,|\A|)$ for order $\leq_{lex}$.

We add in the induction hypothesis that $\M$ has a neutral element $1$ for multiplication.

If $|\M|=1$, i.e. $M=\set{a}$, then $f_a$ is the constant function $0$ or $\infty$, which is $\ltlq$-computable.

If $\A=\set{a}$, we can consider that $\M=\set {a^i:0\leq i\leq p}\cup\set{(a^p)^\sharp}$ (by aperiodicity of $\M$) and $(a^p)^\sharp\leq a^p$ is the only pair in $\leq$.
We can show that for all $b\in\M$, $f_b$ is $\ltlq$-computable:
\begin{itemize}
\item If $i<p$, $f_{a^i}\approx\sem{\bigwedge_{0\leq j<i} X^j a\wedge X^i \Omega}$,
\item $f_{a^p}\approx\sem{\bot\UN\Omega}$,
\item $f_{(a^p)^\sharp}\approx\sem{\bigwedge_{0\leq j<p} X^j a}$
\end{itemize}

Let us assume that $|\M|>1$, $|\A|>1$, and the theorem is true for all $(|\M'|,|\A'|)<_{lex}(|\M|,|\A|)$.

\begin{lem}\label{lem:letb}
There is a letter $b\in\A\setminus\{1\}$ such that there is no $a\in \A\setminus\{b\}$ with $a\leq b$.
\end{lem}
\begin{proof}
We show that $1$ is incomparable with all other elements in $\M$.
We assumed that the order $\leq$ in $\M$ is generated by the axioms of stabilisation monoids.
We recall the relevant axioms here: 
\begin{itemize}
\item[1.] if $e$ is idempotent, we have $e^\sharp\leq e$,
\item[2.] if $e,f$ are idempotents such that $e\leq f$, we have $e^\sharp\leq f^\sharp$,
\item[3.] if $x_1\leq y_1$ and $x_2\leq y_2$ then $x_1x_2\leq y_1y_2$.
\end{itemize}

We show by induction on the length of the derivation that $1$ is not comparable with any other element.
The only rule with no premise is the first one, and it cannot be used to compare $1$ with something else. Indeed, $e=1\Leftrightarrow e^\sharp=1$.

Assume $1\leq m$ or $m\leq 1$ with $m\neq 1$, and consider a derivation of minimal length showing this inequality.

If the last rule of the derivation is rule 2, we have for instance $e^\sharp=1$ and $m=x^\sharp\neq 1$. The premise is $1\leq x$ with $x\neq 1$, contradicting the minimality assumption.

If the last rule of the derivation is rule 3, we have for instance $x_1x_2=1$ and $m=y_1y_2$. Without loss of generality we can assume $y_1\neq 1$. Since $x_1x_2=1$, we get $x_1(x_1x_2)x_2=1$ as well, and more generally $x_1^n x_2^n=1$ for all $n\in\N$. This means $x_1^\omega x_2^\omega=1$. Therefore, $x_1^{\omega+1}x_2^\omega=x_1$. but since $\M$ is aperiodic, we have $x_1^{\omega+1}=x_1^\omega$, and we obtain $x_1=x_1^\omega x_2^\omega=1$.
We get a premise of the form $1\leq y_1$ with $y_1\neq 1$, contradicting again the minimality assumption.

We showed that $1$ is incomparable with all elements in $\M$. Therefore, it suffices to take for $b$ any minimal element of $\A\setminus\{1\}$ to obtain the wanted result.
\end{proof}

We choose a letter $b\neq 1$ given by Lemma \ref{lem:letb}. Let $\B=\A\setminus\set{b}$.

In the following, we will use the notation $\neg b$ as a shortcut for $\bigvee_{a\in \A\setminus\{b\}}$. Notice that in general, we could still have $\sem{\neg b}(b)=0$, if there was a letter $a<b$, but the choice of $b$ prevents this eventuality, and justifies the notation $\neg b$.

Let $L_0=\B^*$, $L_1=\B^*b\B^*$, and $L_2=\B^*b(\B^*b)^+\B^*$.
We have $\A^*=L_0\cup L_1\cup L_2$.

We define restrictions of $f_m$: $f_0,f_1,f_2$ on $L_0,L_1,L_2$ respectively (giving value $\infty$ outside of the domain). We have $f_m=\min(f_0,f_1,f_2)$.
Hence it suffices to show that the $f_i$'s are $\ltlq$-computable to get that $f_m$ is also $\ltlq$-computable (always up to $\approx$).

$f_0$ is computed by $\M$ on alphabet $\B$, so by induction hypothesis there is a formula $\varphi_0$ on $\B$ computing $f_0$.
The formula $\varphi_0'=\varphi_0\wedge G\neg b$ is a formula on $\A$ computing $f_0$.

%In the following, we will write $\phi(u)$ instead of $\sem{\phi}(u)$ for lisibility.
For all $x\in\M$, let $\varphi_x$ be the $\ltlq$-formula on $\B$ computing $f_x$ (restricted to $\B^*$), these formulae exist by induction hypothesis, since $|\B|<|\A|$.

If $\varphi$ is an $\ltlq$-formula on $\B$, we define its ``relativisation'' $\varphi'$ on $\A$ which has the effect of $\varphi$ on the part before $b$ in a word. 
We define $\varphi'$ by induction in the following way:
$$\begin{array}{lcl}
a'&=&a\wedge XFb\\
\Omega'&=&b\\
(\varphi\wedge\psi)'&=&\varphi'\wedge\psi'\\
(\varphi\vee\psi)'&=&\varphi'\vee\psi'\\
(X\varphi)'&=&X\varphi'\wedge\neg b\\
(\varphi U\psi)'&=&(\varphi'\wedge\neg b)U\psi'\\
(\varphi \UN\psi)'&=&(\varphi'\wedge\neg b)\UN\psi'\\
\end{array}$$
With this definition, $\sem{\varphi'}(u_1bu_2)=\sem{\varphi}(u_1)$ for any $u_1\in\B^*$ and $u_2\in\A^*$.

We define the following formula on $\A$:

$$\varphi_1=\big(\bigvee_{xby=m} (\varphi_x'\wedge F(b\wedge X\varphi_y))\big)\wedge (\neg b U(b\wedge X G\neg b))$$

The second part controls that the word is in $L_1$.
We show $\sem{\varphi_1}\approx f_1$.

Let $u\in L_1$, we can write $u=u_1bu_2$ with $u_1,u_2\in\B^*$.\\
By definition of $\varphi_1$,\\
$\begin{array}{ll}
\sem{\varphi_1}(u)&=\min_{xby=m}\max(\sem{\varphi_x'}(u),\sem{\varphi_y}(u_2))\\
&=\min_{xby=m}\max(\sem{\varphi_x}(u_1),\sem{\varphi_y}(u_2))\\
&=\min_{xby=m}\max(f_x(u_1),f_y(u_2)).
\end{array}$

We have for any $z\in\M$ and $v\in\B^*$, $\rho(v)\succeq\bot|_{f_z(v)}z$ where $\bot$ is an extra smallest element (by definition of $f_z$).

But for any $x,y$ such that $xby=m$,\\
$\begin{array}{ll}
\rho(u)&\sim\tilde\rho(\rho(u_1)b\rho(u_2))\\
&\succeq\tilde\rho(\bot|_{f_x(u_1)}x\cdot b\cdot \bot|_{f_y(u_2)}y)\\
&\succeq \bot|_{\max(f_x(u_1),f_y(u_2)}m.
\end{array}$

It implies that for some $\beta$ (not depending on $u$), $\forall x,y$ such that $xby=m$, $f_m(u)\leq_\beta\max(f_x(u_1),f_y(u_2))$. 

In particular, $f_1(u)=f_m(u)\leq_\beta\min_{xby\in I}\max(f_x(u_1),f_y(u_2))=\sem{\varphi_1}(u)$.
We can conclude $f_1\preccurlyeq\sem{\varphi_1}$.

Conversely, let us assume that $f_1(u)\leq n$, it means that $\rho(u)(n)\geq m$.
but $\rho(u)\sim_\alpha\rho(u_1)\cdot b\cdot\rho(u_2)$, so $\rho(u_1)(\alpha(n))\cdot b\cdot\rho(u_2)(\alpha(n))\geq m$.

Let $x=\rho(u_1)(\alpha(n))$ and $y=\rho(u_2)(\alpha(n))$, we have $f_x(u_1)\leq\alpha(n)$ and $f_y(u_2)\leq\alpha(n)$, so $\max(f_x(u_1),f_y(u_2))\leq\alpha(n)$. We get $\sem{\varphi_1}(u)\leq\alpha(n)$, and in conclusion $\sem{\varphi_1}\preccurlyeq f_1$.
This concludes the proof of $\sem{\varphi_1}\approx f_1$.
\bigskip

Last but not least, we have to show that $f_2$ is $\ltlq$-computable up to $\approx$.
For that we will finally use the induction hypothesis on the size of the monoid (until now we only have decreased the size of the alphabet and kept the monoid unchanged).

We define the stabilization monoid $\M'=\perm{Mb\cap bM,\circ,\natural,\leq'}$ in the following way:
$xb\circ by=xby$, and for $xb$ idempotent $(xb)^\natural=(x^\omega)^\sharp b$ where $x^\omega=x^{|\M|}$ is idempotent, since $\M$ is aperiodic.
This monoid generalizes the construction of \emph{local divisor} from \cite{DiekertG}.
%The order $\leq'$ 
$\M'$ is a stabilization monoid, let $\rho'$ be compatible with $\M'$.
We can first notice that this definition implies that for all $k\in\N$, $(xb)^k=x^kb$, so $\M'$ is also aperiodic.
Moreover, we show that $1\notin\M$:
Assume $1\in\M'$, let $k=|\M|$, $1=xb=(xb)^k=x^kb^k=x^kb^{k+1}=(xb)^kb=1b=b$, but $b\neq 1$ so $1\notin\M'$.
Remark that $b$ is the neutral element for $\circ$ in $\M'$, and $|\M'|<|\M|$, which allows us to use induction hypothesis on $\M'$ with alphabet $\M'$.

Let $\Delta=b(\B^*b)^+$, then $L_2=\B^*\Delta\B^*$.

\noindent
Let $d\in\M$, we first want to show that $f_d$ over language $\Delta$ is $\ltlq$-computable up to $\approx$.

Let $\sigma:\Delta\rightarrow(\M'^\N)^*$ defined by
$$\sigma(bu_1b\dots u_kb)=(b\rho(u_1)b)\dots(b\rho(u_k)b).$$

By induction hypothesis, for any $x\in\M'$, there exists an $\ltlq$-formula $\psi_x$ on alphabet $\M'$ and a correction function $\alpha$ such that for any $v\in\M'^*$,\\ $\sem{\psi_x}(v)\approx_\alpha\inf\set{n\in\N:\rho'(v)(n)\geq x}$.\\
\begin{defi}
Let $\semi$ be a stabilization monoid. Let $f$ be a function $\semi^*\rightarrow\N^\infty$, and $\semi^\uparrow$ be the set of $\alpha$-increasing sequences of elements of $\semi$ (for some $\alpha$). We define $\tilde f: \semi^\uparrow\rightarrow\N_\infty$ by $\tilde f(\vec u)=\inf\set{n:f(u_n)\leq n}$.
\end{defi}
Remark that this notation is coherent with the $\tilde{}$ operator previously defined for functions $S^+\to \N\to S$ in the sense that if $f^\approx$ is recognized by $\semi,h,I$ with compatible function $\rho$, i.e.  $f\approx u \mapsto I[\rho(h(u))]$, then $\tilde f\approx \vec u \mapsto I[\tilde{\rho}(h(\vec u))]$.

This definition is needed because we already make use of $\rho$ to define $\sigma$, so each word of $\B^*$ is mapped to a sequence of elements. However we will need to recombine these various elements, so we will need a formula which is able to take as input sequences instead of words. This will be obtained by applying the tilde operator to the semantic of a formula.

\begin{lem}\label{tech}
We claim that there exists $\alpha$ and $\phi_d$ an $\ltlq$-formula on alphabet $\A$ such that for all $u\in\Delta$ and $v\in\B^*$:
%\begin{itemize}
%\item[(a)]$\tilde\rho'(\sigma(u))\sim_\alpha\rho(u)$,
$$\sem{\phi_d}(uv)\approx_\alpha\widetilde{\sem{\psi_d}}(\sigma(u))\approx_\alpha f_d(u)$$
%\end{itemize}
%and consequently $\sem{\phi_d}(uv)\approx_{\alpha^2} f_d(u)$
\end{lem}

Intuitively, $\phi_d$  forgets the last $\B^*$-component $v$ of its input, and is able to apply $\sigma(u)$ to split the word according to the $b$'s, and compute the value of each component with respect to $\psi_d$.

With this result we can build a formula $\varphi_2$ computing $f_2$:
$$\varphi_2=(\bigvee_{xdy=m}(\varphi_x'\wedge F(b\wedge X\phi_d))\wedge F(b\wedge X(G\neg b\wedge \varphi_y)))\wedge \varphi_{L_2}$$
where $\varphi_{L_2}=F(b\wedge X Fb)$ controls that the word is in $L_2$.

By construction, lemmas and induction hypothesis, there exists $\alpha$ such that for all $v_1,v_2\in\B^*$ and $u\in\Delta$,\\
$\begin{array}{ll}
\sem{\varphi_2}(v_1uv_2)&\approx_\alpha\min_{xdy=m}\max(\sem{\varphi'_x}(v_1uv_2),\sem{\phi_d}(uv_2),\sem{\varphi_y}(v_2))\\
&\approx_\alpha\min_{xdy=m}\max(f_x(v_1),f_d(u),f_y(v_2)).
\end{array}
$

The proof that $min_{xdy=m}\max(f_x(v_1),f_d(u),f_y(v_2))\approx f_m(v_1uv_2)$ is similar to the proof of $\sem{\varphi_1}\approx f_1$. 

All this together gives us $\sem{\varphi_2}\approx f_2$, which concludes the proof.

\end{proof}

\subsubsection*{Proof of Lemma \ref{tech}}
\begin{proof}
First let us show that $\widetilde{\sem{\psi_d}}(\sigma(u))\approx_{\alpha} f_d(u)$ for some $\alpha$ and all $u\in\Delta$.
Let $u=bu_1bu_2\dots u_kb$ with $u_i\in\B^*$.
For each $i\in\sem{1,k}$ and $t\in\N$, $\rho(u_i)(t)=a_{i,t}\in\M$.
For all $t\in\N$, let $v_t=(ba_{1,t}b)\dots(ba_{k,t}b)$, $v_t$ is a word on $\M'$ of length $k$, and $\sigma(u)=(v_t)_{t\in\N}$.
Finally, let $w_t=ba_{1,t}ba_{2,t}\dots ba_{k,t}b$ of length $2k+1$ on $\M$.

We have:

$\begin{array}{ll}
\widetilde{\sem{\psi_d}}(\sigma(u)))&=\inf\set{t:\sem{\psi_d}(v_t)\leq t}\\
&\approx\inf\set{t:\inf\set{n:\rho'(v_t)(n)\geq d}\leq t}\\
\end{array}$
%
%On peut vérifier que $\rho'(v_t)\sim\rho(w_t)$ pour tout $t$ : étant donné un $n$-calculs sur $v_t$, on peut le voir comme un $n$-calcul sur $w_t$.
%
%Pour ce faire, on peut montrer que pour tout

We will show that $\rho'(v_t)\sim\rho(w_t)$ for all $t$.
It suffices to verify that $\rho'$ and $\rho$ both verify all axioms of Theorem \ref{compaxioms} over $(ba_1b)\dots(ba_kb)$ and $ba_1ba_2\dots a_kb$ respectively. Let $\alpha$ and $\alpha'$ be the correction functions given by this theorem for $\rho$ and $\rho'$.

\begin{description}
%\item[Monotonicity.] $\rho$ is~$\alpha$-monotone,
\item[Letter.] For all $a\in M$, we have $\rho'(bab)\sim_{\alpha'}bab\sim_\alpha\rho(bab)$,
\item[Product.] For all $a_1,a_2\in M$, we have $\rho'((ba_1b)(ba_2b))\sim_{\alpha'} (ba_1b)\circ(ba_2b) = ba_1ba_2b \sim_{\alpha^4}\rho(ba_1ba_2b)$,
\item[Stabilization.] Let $bab$ be an idempotent of $\M'$, and $m\in\N$, 
Notice that $babab=bab$, so for all $l\geq 1$, $(ba)^lb=bab$, where product is relative to $\M$.

$\rho'((bab)^m)\sim_{\alpha'} (bab)^\natural|_m (bab)=(ba)^{\omega\sharp}b|_m (bab)$.
We perform the euclidean division $m=|\M|m'+m''$ avec $m''<|\M|$,\\
$\begin{array}{ll}
\rho((ba)^mb)&\sim \rho(((ba)^{|\M|})^{m'})\cdot (ba)^{m''}b\\
&\sim^{(1)} \rho({((ba)^{\omega})}^{m'})\cdot  (bab)\\
&\sim (ba)^{\omega\sharp} (bab)|_{m'}(ba)^{\omega} (bab)\\
&\sim^{(2)} (ba)^{\omega\sharp}b|_m (bab).
\end{array}$

Equivalence $(1)$ is obtained by aperiodicity of $\M$ ($(ba)^\omega$ is a letter here, not a word of length $|\M|$), and equivalence $(2)$ by using the property $m\approx_{\times(|\M|+1)}m'$.
\end{description}

These three cases show that any $n$-tree in $\M'$ over $v_t$ can be transformed into an $n$-tree over $w_t$, since each type of node is preserved. The substitution property corresponds to branching several $n$-trees together, so it is necessary to treat it here.

Thus we obtain $\rho'(v_t)\sim\rho(w_t)$ for all $t$.

%
%We can verify that $\rho'(v_t)\sim\rho(w_t)$ for any $t$: we check that $\rho'$ verifies the same axioms on words $(ba_1b)\dots(ba_kb)$ than $\rho$ does for $ba_1ba_2\dots a_kb$. The only interesting case is the stabilization rule: let $bab$ be an idempotent of $\M'$ and $m\in\N$, 
%$\rho'((bab)^m)\sim (bab)^\natural|_m (bab)\sim (ba)^{\omega\sharp}b|_m (bab)$.
%But if $m=|\M|m'+m''$ with $m''<|\M|$,\\
%$\begin{array}{ll}
%\rho((ba)^mb)&\sim \rho((ba)^{|\M|})^{m'})\cdot (ba)^{m''}b\\
%&\sim^{(1)} \rho({((ba)^{\omega})}^{m'})\cdot  (ba)^{m''}b\\
%&\sim (ba)^{\omega\sharp} (ba)^{m''}b|_{m'}(ba)^{\omega} (ba)^{m''}b\\
%&\sim^{(2)} (ba)^{\omega\sharp}b|_m (bab).
%\end{array}$
%
%We get the equivalence $(1)$ by aperiodicity of $\M$ ($(ba)^\omega$ is now a letter and no longer a word of length $|\M|$), and $(2)$ by the fact that $bab$ is idempotent in $\M'$ so $(ba)^\omega (ba)^{m''}b=bab$, and $(ba)^{\omega} (ba)^{m''}=(ba)^{\omega}$ by aperiodicity of $\M$ (and also $m\approx_{\times(|\M|+1)}m'$).
%
 %We can then apply the Theorem \ref{unirho}: $\rho$ is unique up to $\sim$, hence we have $\rho'(v_t)\sim\rho(w_t)$ for any $t$.

Moreover, let $\vec w=(w_t)_{t\in\N}$, we want to show the following property:
$$(EQ):~\inf\set{n':\tilde\rho(\vec w)(n')\geq d}\approx\inf\set{t:\inf\set{n:\rho(w_t)(n)\geq d}\leq t}.$$
Let $N'=\inf\set{n':\tilde\rho(\vec w)(n')\geq d}$. Notice that $\rho(w_{N'})(N')\geq d$, 
so $N'\geq\inf\set{t:\inf\set{n:\rho(w_t)(n)\geq d}\leq t}$.

Conversely, let $T=\inf\set{t:\inf\set{n:\rho(w_t)(n)\geq d}\leq t}$ and $N$ the corresponding value of $\inf\set{n:\rho(w_t)(n)\geq d}$, we have $N\leq T$ and $\rho(w_t)$ is $\alpha$-increasing, so $\rho(w_T)(T)\geq_\alpha \rho(w_T)(N)\geq d$, i.e. $T\geq_\alpha\inf\set{n':\tilde\rho(\vec w)(n')\geq d}$.

Hence we have the equivalence $(EQ)$.

Finally,\\
$\begin{array}{llr}
\widetilde{\sem{\psi_d}}(\sigma(u)))&\approx\inf\set{t:\inf\set{n:\rho(w_t)(n)\geq d}\leq t}&\\
&\approx\inf\set{n:\tilde\rho(\vec w)(n)\geq d}&\text{   by }(EQ)\\
&=\inf\set{n:\tilde\rho(b\rho(u_1)b\rho(u_2)\dots\rho(u_k)b)(n)\geq d}&\\
&\approx\inf\set{n:\rho(bu_1bu_2 \dots u_kb)(n)\geq d}&\text{  Substitution axiom}\\
&\approx f_d(u).
\end{array}$

which concludes the proof of $\widetilde{\sem{\psi_d}}(\sigma(u)))\approx f_d(u)$.
\bigskip

It remains to show that there exists a formula $\phi_d$ and an $\alpha$ such that  for all $u,v\in \Delta\times\B^*$, $\sem{\phi_d}(uv)\approx_\alpha\widetilde{\sem{\psi_d}}(\sigma(u))$.

If $\psi$ is an $\ltlq$-formula on $\M'$, we define $\psi^\bigstar$ on alphabet $\A$ by induction on $\psi$:

$\begin{array}{ll}
x^\bigstar&=(b\wedge XFb)\wedge(X\varphi_x')\\
(\psi_1\wedge\psi_2)^\bigstar&=\psi_1^\bigstar\wedge\psi_2^\bigstar\\
(\psi_1\vee\psi_2)^\bigstar&=\psi_1^\bigstar\vee\psi_2^\bigstar\\
(X\psi)^\bigstar&=\neg b U(b\wedge\psi^\bigstar)\\
(\psi_1 U\psi_2)^\bigstar&=(b\implies\psi_1^\bigstar)U(b\wedge\psi_2^\bigstar)\\
(\psi_1 \UN\psi_2)^\bigstar&=(b\implies\psi_1^\bigstar)\UN(b\wedge\psi_2^\bigstar).
\end{array}$

Where $\varphi_x'$ is defined as before for any $\varphi_x$ on alphabet $\B$.

Let us show by induction on $\psi$ that that $\sem{\psi^\bigstar}(uv)\approx\widetilde{\sem\psi}(\sigma(u))$ for $u=bu_1bu_2\dots u_kb\in\Delta$ and $v\in\B^*$:
\begin{itemize}
\item If $x\in\M'$,\\
 $\sem{x^\bigstar}(uv)=\sem{\varphi_x'}(u_1bu_2\dots u_kbv)=\sem{\varphi_x}(u_1)$, and\\
$\widetilde{\sem{x}}(\sigma(u))=\inf\set{n:\sem{x}(\rho(u_1)(n))\leq n}\approx\inf\set{n:(\rho(u_1)(n))\geq x}\approx\sem{\varphi_x}(u_1)$.
\item $\wedge$ case:

$\begin{array}{ll}
\sem{(\psi_1\wedge\psi_2)^\bigstar}(uv)&=\max(\sem{\psi_1^\bigstar}(uv),\sem{\psi_2^\bigstar}(uv))\\
&\approx\max(\widetilde{\sem{\psi_1}}(\sigma(u)),\widetilde{\sem{\psi_2}}(\sigma(u)))\\
&\approx\widetilde{\sem{\psi_1\wedge\psi_2}}(\sigma(u))
\end{array}$
\item $\vee$ case:

$\begin{array}{ll}
\sem{(\psi_1\vee\psi_2)^\bigstar}(uv)&=\min(\sem{\psi_1^\bigstar}(uv),\sem{\psi_2^\bigstar}(uv))\\
&\approx\min(\widetilde{\sem{\psi_1}}(\sigma(u)),\widetilde{\sem{\psi_2}}(\sigma(u)))\\
&\approx\widetilde{\sem{\psi_1\vee\psi_2}}(\sigma(u))
\end{array}$
\item $X$ case:

$\begin{array}{ll}
\sem{(X\psi)^\bigstar}(uv)&=\sem{\psi^\bigstar}(bu_2b\dots u_kbv)\\
&\approx\widetilde{\sem\psi}(\sigma(bu_2b\dots u_kb))\\
&\approx\widetilde{\sem{X\psi}}(\sigma(bu_1bu_2b\dots u_kb))
\end{array}$
\item $U$ case:

$\begin{array}{ll}
\sem{(\psi_1 U\psi_2)^\bigstar}(uv)&=\min_{1\leq j\leq k}(\max(\sem{\psi_2^\bigstar}(bu_jb\dots u_kbv),\max_{1\leq i\leq j}\sem{\psi_1^\bigstar}(bu_ib\dots u_kbv)))\\
&\approx\min_{1\leq j\leq k}(\max(\widetilde{\sem{\psi_2}}(\sigma(bu_jb\dots u_kb)),\max_{1\leq i\leq j}\widetilde{\sem{\psi_1}}(\sigma(bu_ib\dots u_kb))))\\
&\approx\widetilde{\sem{\psi_1 U \psi_2}}(\sigma(u))\\
\end{array}$

\item The $\UN$ case is the same than above, allowing at most $N$ mistakes for $\psi_1$.

We now just have to take $\phi_d=\psi_d^\bigstar$ to complete the proof of Lemma \ref{tech}.
\end{itemize}
\end{proof}

\begin{cor}
The class of $\ltlq$-definable cost functions is decidable.
\end{cor}
\begin{proof}
Theorems \ref{ltlap} and \ref{apltl} imply that it is equivalent for a regular cost function to be $\ltlq$-definable or to have an aperiodic syntactic stabilization semigroup.
If $F$ is given by an automaton or a stabilization semigroup, we can compute its syntactic stabilization semigroup $\semi_F$ (see \cite{CKL}) and decide if $F$ is $\ltlq$-definable by testing aperiodicity of $\semi_F$. This can be done simply by iterating at most $|\semi_F|$ times all elements of $\semi_F$ and see if each element $a$ reaches an element $a^k$ such that $a^{k+1}=a^k$.
\end{proof}

\section{Conclusion}
We first defined $\ltlq$ as a quantitative extension of LTL. We started the study of $\ltlq$ by giving an explicit translation from $\ltlq$-formulae to $B$-automata and $S$-automata, therefore showing that the boundedness (and comparison) problem for $\ltlq$-formulae is PSPACE-complete.
 We then showed that the expressive power of $\ltlq$ in terms of cost functions is the same as aperiodic stabilization semigroups. The proof uses a new syntactic congruence, which has a general interest in the study of regular cost functions. This result implies the decidability of the $\ltlq$-definable class of cost functions.
%
%We first defined \ose s as a "typed" extension of \se s.
%We then defined a syntactic congruence for cost functions, which allows to describe the minimal stabilization semigroup of a regular cost functions in terms of equivalence class of \ose s.
\bigskip

As a further work, we can try to put \ose s in a larger framework, by doing an axiomatization of $\omega\sharp$-semigroups.
We can also extend this work to infinite words, and define an analog to Büchi automata for cost functions. To continue the analogy with classic languages results, we can define a quantitative extension of FO describing the same class as $\ltlq$, and search for analog definitions of counter-free $B$-automata and star-free $B$-regular expressions.
The translation from $\ltlq$-formulae to $B$-automata can be further studied in terms of optimality of number of counters of the resulting $B$-automaton.

\subsection*{Acknowledgments}
I am very grateful to my advisor Thomas Colcombet for our helpful discussions, and for the guidelines he gave me on this work, and to Michael Vanden Boom for helping me with language and presentation issues. I also thank the anonymous reviewers for their useful comments on the presentation.

\bibliographystyle{alpha}
\bibliography{biblio}
%\newpage
%\input{appendix.tex}

\end{document}